%% file: main.tex
\documentclass[a4paper,onecolumn,superscriptaddress,nofootinbib,10pt,floatfix,accepted=2022-09-05]{quantumarticle}
\pdfoutput=1

\usepackage[utf8]{inputenc}
\usepackage[english]{babel}
\usepackage[T1]{fontenc}
\usepackage{amsmath} 				
\usepackage{amssymb} 				
\usepackage{amsthm}
\usepackage{array}       				
\usepackage{graphicx}  					
\usepackage{xcolor}       			
\usepackage{url}     
\usepackage{cite}
\usepackage{adjustbox}
\usepackage[shortlabels]{enumitem}
\usepackage{braket}
\usepackage{mathtools}
\usepackage{enumitem}
\usepackage{tikz}
\usetikzlibrary{positioning}
\usetikzlibrary{calc}
\usetikzlibrary{shapes.multipart}
\definecolor{colorRE}{HTML}{E76F51}
\definecolor{colorYE}{HTML}{E9C46A}
\definecolor{colorGR}{HTML}{2A9D8F}
\usepackage{hyperref}
\newcolumntype{C}[1]{>{\centering\let\newline\\\arraybackslash\hspace{0pt}}m{#1}}

\theoremstyle{definition}
\newtheorem{definition}{Definition}

\newtheorem{algorithm}{Algorithm}

\theoremstyle{plain}
\newtheorem{theorem}{Theorem}%[section]
\newtheorem{corollary}{Corollary}
\newtheorem{lemma}{Lemma}

\newcommand{\abs}[1]{\lvert #1 \rvert}
\newcommand{\norm}[1]{\lVert #1 \rVert}
\newcommand{\Abs}[1]{\left\lvert #1 \right\rvert}
\newcommand{\Norm}[1]{\left\lVert #1 \right\rVert}
\newcommand{\rest}{\mathrm{\textbf{R}}}
\newcommand{\trace}[1]{\mathrm{tr}\left(#1\right)}

\newcommand{\id}{\mathbb{I}}
\newcommand{\ee}{\mathrm{e}}
\newcommand{\ii}{\mathrm{i}}
\newcommand{\cc}{\mathbb{C}}
\newcommand{\EE}{\mathbb{E}}
\newcommand{\mL}{\mathcal{L}}
\newcommand{\mV}{\mathcal{V}}

\newcommand{\nn}{\boldsymbol{\nu}}
\newcommand{\bb}{\boldsymbol{b}}

\begin{document}

\title{Randomizing multi-product formulas for Hamiltonian simulation}

	\author{\href{https://orcid.org/0000-0002-8706-1732}{Paul\ K.\ Faehrmann}}
	\thanks{\href{mailto:paul.faehrmann@fu-berlin.de}{paul.faehrmann@fu-berlin.de}\newline P.K.F.~and
M.S.~have contributed equally.}
	\affiliation{Dahlem Center for Complex Quantum Systems, Freie Universität Berlin, 14195 Berlin, Germany}

	\author{\href{https://orcid.org/0000-0002-6419-302X}{Mark\ Steudtner}}
	\affiliation{Dahlem Center for Complex Quantum Systems, Freie Universität Berlin, 14195 Berlin, Germany}
    \author{\href{https://orcid.org/0000-0002-8291-648X}{Richard\ Kueng}}
	\affiliation{Institute for Integrated Circuits, Johannes Kepler University Linz, Austria}

	\author{\href{https://orcid.org/0000-0002-0749-8126}{Mária\ Kieferová}}
    \affiliation{Centre for Quantum Computation and Communication Technology,
Centre for Quantum Software and Information,
University of Technology Sydney,
NSW 2007, Australia}
	
	\author{\href{https://orcid.org/0000-0003-3033-1292}{Jens\ Eisert}}
	\affiliation{Dahlem Center for Complex Quantum Systems, Freie Universität Berlin, 14195 Berlin, Germany}
	\affiliation{Helmholtz-Zentrum Berlin f{\"u}r Materialien und Energie, Hahn-Meitner-Platz 1, 14109 Berlin, Germany}

% \date{\today}

\begin{abstract}
Quantum simulation, the simulation of quantum processes on quantum computers, suggests a path forward for the efficient simulation of problems in condensed-matter physics, quantum chemistry, and materials science. 
While the majority of quantum simulation algorithms are deterministic, a recent surge of ideas has shown that randomization can greatly benefit algorithmic performance. 
In this work, we introduce a scheme for quantum simulation that unites the advantages of randomized compiling on the one hand and higher-order multi-product formulas, as they are used for example in linear-combination-of-unitaries (LCU) algorithms or quantum error mitigation, on the other hand. In doing so, we propose a framework of randomized sampling that is expected to be useful for programmable quantum simulators and present two new multi-product formula algorithms tailored to it. 
Our framework reduces the circuit depth by circumventing the need for oblivious amplitude amplification required by the implementation of multi-product formulas using standard LCU methods, rendering it especially useful for early quantum computers used to estimate the dynamics of quantum systems instead of performing full-fledged quantum phase estimation. 
Our algorithms achieve a simulation error that shrinks exponentially with the circuit depth. To corroborate their functioning, we prove rigorous performance bounds as well as the concentration of the randomized sampling procedure. 
We demonstrate the functioning of the approach for several physically meaningful examples of Hamiltonians, including fermionic systems and the Sachdev–Ye–Kitaev model, for which the method provides a favorable scaling in the effort.
\end{abstract}
\maketitle

\section{Introduction}
The simulation of quantum processes on quantum computers is one of the most eagerly anticipated use cases for quantum computing. The ability to simulate a system's time evolution promises to provide insights into the dynamics of interacting quantum systems in situations where approximate classical simulation methods fail and constitutes one of the cornerstones of quantum technologies \cite{Roadmap}. 

This work aims at improving algorithms that simulate the dynamics of expectation values of observables. 
The need for developing such machinery stems from the observation that
state-of-the-art quantum devices and early quantum computers are still rather limited in their realizable circuit depths and control. We, therefore, assume only access to a quantum-oracle machine that implements single-qubit state preparation, controlled time evolution and quantum measurements and strive for minimizing the required depth of suitable quantum algorithms. Such a setting explicitly allows for the use of product formulas, which are the earliest algorithms proposed for the simulation of time-independent local Hamiltonians~\cite{lloydUniversalQuantumSimulators1996b}.

Such Trotter-Suzuki methods, as they are called, have evolved from comparably simple prescriptions for local Hamiltonians to sophisticated schemes able to capture more general sparse time-independent Hamiltonians~\cite{aharonovAdiabaticQuantumState2003, berry2007efficient, wiebe2010higher} as well as time-dependent Hamiltonians ~\cite{wiebe2011simulating, poulin2011quantum} and open quantum systems~\cite{Kliesch-PRL-2011,sweke2016digital}.
Despite their relative simplicity, product formulas are still at the forefront of Hamiltonian simulation techniques. A numerical study has shown that product formulas can in practice outperform more complex techniques~\cite{childs2018toward} and their complexity is better than initial estimates~\cite{childsTheoryTrotterError2020} suggest. In fact, product formulas are nearly optimal for lattice model simulations~\cite{childs2019nearly}.

Multi-product formulas introduced in the \emph{linear-combinations-of-unitaries (LCU)}~\cite{childsHamiltonianSimulationUsing} approach have been built upon previous results of Trotter-Suzuki methods and have recently been improved~\cite{lowWellconditionedMultiproductHamiltonian2019}. The
idea of linearly combining unitaries has led to quantum algorithms with an exponential speedup in precision~\cite{berry2015hamiltonian, berry2014exponential, berry2015simulating, lowHamiltonianSimulationQubitization2019b} which have been recently
gaining popularity. Besides quantum simulation in the LCU framework, the mathematical construct of multi-product formulas (see Fig.~\ref{fig:mpf_frameworks}) can also be used in quantum error mitigation \cite{endoHybridQuantumClassicalAlgorithms2021}.
Although these methods have optimal asymptotic error scaling, they inherently require deep quantum circuits for their implementation and are thus not suitable for applications in near-term devices or early quantum computers. It is therefore crucial to find algorithms that require shorter circuits and fewer digital gates. 

\begin{figure}[t!]
    \centering
    \begin{adjustbox}{max width=\linewidth}
    \input{figures/multi-product-formulas_tikz.tex}
    \end{adjustbox}
    \caption{Multi-product formulas, although introduced with use in the LCU framework in mind, can also be used in different frameworks, such as quantum error mitigation or randomized sampling, which is the main contribution of this work.}
    \label{fig:mpf_frameworks}
\end{figure}
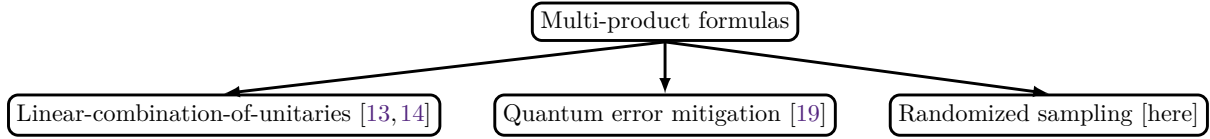

Recently, a new element has been introduced to aid this search: the element of \emph{classical randomness}. 
The idea of randomization in Hamiltonian simulation has heralded a renaissance of product formula methods~\cite{campbellShorterGateSequences2017b,campbellRandomCompilerFast2019a,childsFasterQuantumSimulation2019,ouyangCompilationStochasticHamiltonian2020,chenQuantumSimulationRandomized2020}.
For single steps of such techniques, a rigorous understanding has recently been reached \cite{chenQuantumSimulationRandomized2020}. 
Randomized algorithms can also be considered when one is only interested in estimating expectation values. For such a task we do not need to prepare the time-evolved state (from which the observable will be measured) perfectly.  
Instead, a lower-effort randomized algorithm can be used such that the correct expectation value is obtained only after averaging the measurement outcomes.  

The aims of this work are twofold: First, we strive for combining the advantages of higher-order multi-product formulas with those of schemes of randomized compiling, to create a novel framework in which multi-product formulas can be put to good use, which we dub ``randomized sampling''. 
Specifically, we adopt the above scenario of computing expectation values of observables and propose to sample product formulas from multi-product formulas, by which we circumvent the need to use the LCU framework and its corresponding methods such as block encodings. 
Instead, we implement multi-product formulas on average through random sampling. This results in a notable reduction in algorithmic depth at the expense of additional circuit evaluations. 
In this way, we see how notions of randomized compiling and higher-order multi-product formulas -- when suitably brought together -- allow for more resource-efficient notions of quantum simulation amenable to programmable devices and early quantum computers, with a focus on estimating the dynamics of quantum systems instead of a complete phase estimation procedure. An overview of this framework can be found in Figure~\ref{fig:sampling_algo}.
Second, we develop alternative multi-product formulas tailored to this new framework, which promise to outperform the accuracy of the multi-product formulas introduced by Childs and Wiebe in Ref.~\cite{childsHamiltonianSimulationUsing} in the regime of short simulation times.
We show that using a quantum device limited to the aforementioned operations can already yield improvements to fully analog approaches. It is therefore especially useful in the regime of early quantum computers, when algorithms with a pure focus on noisy-intermediate-scale-quantum devices \cite{preskillQuantumComputingNISQ2018a}, such as variational quantum algorithms, reach their limits, but fully digital algorithms on large, fault-tolerant quantum computers are not yet feasible. The steps required to end up with such a formula are summarized in Figure~\ref{fig:prescription}.
\begin{table}[t!]
    \centering
    \setlength{\tabcolsep}{2pt} % Default value: 6pt
    \renewcommand{\arraystretch}{2} % Default value: 1
    \begin{adjustbox}{max width=\textwidth}
    \begin{tabular}{C{2.95cm}|C{3.4cm}|C{3.15cm}|C{\linewidth-9.5cm}}
    Algorithm & \textbf{Trotter-Suzuki~\cite{suzukiGeneralTheoryFractal1991}}\linebreak(Section~\ref{sec:review_multi_product})  & \textbf{Childs-Wiebe~\cite{childsHamiltonianSimulationUsing}}\linebreak (Section~\ref{sec:review_multi_product}) & \textbf{Novel multi-product formulas [here]}\newline (Section~\ref{ssec:custom_mpfs})\\ \hline
    Formula &  $S_{2\chi}(t/r)^r$ & $\sum_{q=1}^{K+1}C_qS_{2\chi}(t/q)^q$ & 
          $\prod_{r=1}^R\sum_{q=0}^{2\chi R}C_q\!\left(\nn^{(r)}, \bb^{(r)}\right)\; S_{2\chi}\!\left(b_q^{(r)}t\right)$ \\ 
    Free parameters & $r,\,\chi$ & $K,\,\chi$ & $R,\chi,\,\lbrace\bb^{(r)}\rbrace_{r = 1}^{R}$\\
    Dependent \linebreak parameters & - & $\boldsymbol{C}(K,\chi)$ & $\left\lbrace\nn^{(r)}(R,\chi),\, \boldsymbol{C}\!\left(\nn^{(r)}, \bb^{(r)}\right)\right\rbrace_{r=1}^R$\\
    Max. depth & $2Lr\cdot5^{\chi-1}$ & $2L(K+1)\cdot5^{\chi-1}$ & $2LR\cdot5^{\chi-1}$ \\
    Query complexity per sample & $\mathcal{O}\left(\frac{\left(\Lambda t\right)^{1 + 1/(2\chi)}}{\varepsilon^{1/(2\chi)}}\cdot5^{\chi-1}\right) $&$ \mathcal{O}\left(\frac{\ln(\varepsilon)}{\ln(\Lambda t)}5^{\chi-1}\right)$ & $\mathcal{O}\left(\frac{\ln(\varepsilon)}{\ln(\Lambda t)}\frac{5^{\chi-1}}{\chi}\right)$\\
    Error scaling & $\mathcal{O}\left(r(\Lambda t/r)^{2\chi+1}\right)$ & $\mathcal{O}\left((\Lambda t)^{2(\chi+K)+1}\right)$ & $\mathcal{O}\left((\Lambda t)^{2\chi R+1}\right)$ \\ 
    Sampling overhead & $1$ & $\left(\sum_{q=1}^{K+1}\left|C_q\right|\right)^2$ & $\left(\prod_{r=1}^R\sum_{q=0}^{2\chi R}\Abs{C_q\!\left(\nn^{(r)}, \bb^{(r)}\!\right)}\right)^2$
    \end{tabular}
    \end{adjustbox}
    \caption{Comparison of the standard Trotter-Suzuki formula, the multi-product formula used by Childs and Wiebe and the closed-form/matching multi-product formula introduced in this work for approximating the time evolution operator $\exp(-\ii Ht)$ of a Hamiltonian $H=\sum_{k=1}^Lh_k$. We use $\Lambda = \sum_k\Norm{h_k}$ 
    and provide detailed error bounds in Section~\ref{sec:review_multi_product} and Theorem~\ref{theo:error_bound}.
    Each algorithm provides the choice of some free parameters, dictating their performance and resource requirements. 
    The multi-product formulas proposed by Childs and Wiebe in Ref.~\cite{childsHamiltonianSimulationUsing} and those introduced here further come with additional parameters depending on the choice of these free parameters.
    Since we aim for a randomized implementation using the framework introduced in Algorithm~\ref{algo:random_sampling}, the maximal circuit depths, query complexities, and error scalings are slightly different from their implementation in a deterministic, coherent fashion. Furthermore, implementing them in the randomized sampling framework comes with the stated sampling overheads. 
    Note that their errors can also exhibit a commutator scaling as discussed in Refs.~\cite{childsTheoryTrotterError2020,lowWellconditionedMultiproductHamiltonian2019} for all of these methods.}
    \label{tab:compare_methods}
\end{table}

A direct comparison between these newly developed multi-product formulas with those of Childs and Wiebe and Trotter-Suzuki product formulas is presented in Table~\ref{tab:compare_methods}.

The remainder of this work is organized as follows: After a short review of multi-product formulas in Section~\ref{sec:review_multi_product}, we present the framework of randomized sampling and a summary of our main results in Section~\ref{sec:randomized_sampling_main_res}. Section~\ref{sec:proofs} then contains a detailed analysis of the proposed randomized sampling framework and multi-product formulas and gives the proofs of the results discussed in the previous section. To conclude, we compare the performance of our formulas to that of Trotter-Suzuki product formulas and Childs and Wiebe type multi-product formulas \cite{childsHamiltonianSimulationUsing} in Section~\ref{sec:numerics}, discussing both their error bounds and actual performance on several physically plausible and interesting Hamiltonian models of strongly correlated quantum systems, for which we find a favorable performance over known schemes of quantum simulation.

\section{Multi-product formulas}
\label{sec:review_multi_product}
Multi-product formulas constitute the main ingredient of our work, and hence we will briefly review the underlying ideas, starting with product formulas. 
The goal of quantum simulations is to approximate the
quantum dynamics of a complex quantum system described by a many-body Hamiltonian decomposed as
\begin{equation}
    H = \sum_{k=1}^L h_k
\end{equation}
composed of Hermitian 
Hamiltonian 
terms $\{h_k\}$ of neither necessarily
small nor geometrically local support 
defined on a quantum lattice
equipped with a Hilbert space ${\cal H}= (\cc^d)^{\otimes n}$.

Here, $n$ is 
the number of degrees of freedom of finite local dimension $d$. To access the Hamiltonian, we assume to have oracles
$\mathbb{O}_k(t)$ implementing time evolutions under each term in the Hamiltonian
\begin{equation}
    \mathbb{O}_k(t) \ket{\psi} = e^{-\ii h_kt}\ket{\psi} \label{oracles}
\end{equation}
for any $k\in[1,L]$ and $t\in\mathbb{R}$. In digital quantum simulation, these oracles are built from Clifford gates and phase rotations, but we do not assume that such a decomposition is available to us, as the oracles could be implemented by the time evolution of a programmable device. We only require to have control over the implementation, i.e., that we can apply the oracle $\mathbb{O}_k(t)$ depending on the state of a subsystem encoded as a qubit as
\begin{align}
    \ket{0}\!\!\bra{0}\otimes \mathbb{I} + \ket{1}\!\!\bra{1}\otimes \mathbb{O}_k(t)  \, .
\end{align}
 Now, quantum simulation 
aims at making
reliable predictions of the expectation values of
observables $O$ (which in the ideal case are local with a small support
on the lattice, but again, this is 
not a
necessity) at later times $T>0$
\begin{equation}
    \langle O(T)\rangle\coloneqq{\rm tr}(U(T) \rho U^\dagger(T) O)
\end{equation}
for initial quantum states $\rho$, where $U$ is the time evolution operator. 
In practice, this commonly
means to establish ways to 
 approximate the corresponding time evolution operator
\begin{equation}
    U(T) = \ee^{-\ii H T} \, .
\end{equation}
Although it is worth stressing that 
a-priori information about the time $T$ and
properties of the initial state
$\rho$, as well as features of the locality of the underlying Hamiltonian $H$ can be exploited to come up with highly specialized approximations, we do not assume any underlying structure or limitations in this work.

A simple, yet effective approach to approximating the time evolution operator is that of using product formulas, which make use of $N$ consecutive evolutions of individual Hamiltonian terms $h_{k_j}$ by associated short time intervals $\alpha_j t$, defined with
real coefficients $\alpha_j$ such that $\sum_j \alpha_j = 1$. Partial backward evolutions are explicitly allowed, i.e., $\alpha_j<0$ for some $j$, even though $t>0$. In general, such a formula is defined as a product of time evolutions

\begin{equation} \label{eq:proform}
    \ee^{-\ii H t} \approx \prod_{j=1}^N \ee^{-\ii  \alpha_j \, h_{k_j} t} \, 
\end{equation}
where each single term evolution can be implemented via $\mathbb{O}_k(t)$.
Here, we distinguish between a direct evolution by time $t$ and a repeated evolution of short time slices $t/r<1$ such that 
\begin{equation}
\label{eq:reps}
    \ee^{-\ii H t} \approx \left(\prod_{j=1}^N \ee^{-\ii \alpha_j\,  h_{k_j} \frac{t}{r} }\right)^r \, .
\end{equation}
The most accurate known formulas of this type are \emph{Trotter-Suzuki formulas} \cite{suzukiGeneralTheoryFractal1991}. They are recursively defined for any positive integer $\chi$ 
and any time $t$ by
    \begin{gather}
        S_1(t)=\prod_{k=1}^L \ee^{-\ii h_k t} \, ,\\         \label{eq:trotter_suzuki_0}
        S_2(t)=\prod_{k=1}^L \ee^{-\ii h_k t/2}\prod_{k=L}^1 \ee^{-\ii h_k t/2}\, ,\\
        \label{eq:trotter_suzuki}
        S_{2\chi}(t)=\Big(S_{2\chi-2}(s_{2\chi-2}t)\Big)^2 S_{2\chi-2}\Big([1-4s_{2\chi-2}]t\Big) \Big(S_{2\chi-2}(s_{2\chi-2}t)\Big)^2
    \end{gather}
with $s_{2\chi} \coloneqq (4-4^{1/(2\chi+1)})^{-1}$ for any positive integer $\chi$. This specific choice of $s_{2\chi}$ ensures that the Taylor series of $S_{2\chi}(t)$ matches that of the actual time evolution $U(t)$ up to $\mathcal{O}(t^{2\chi+1})$, which makes it a good approximation for $t\ll 1$. 
It is important to stress that constructing an $\mathcal{O}(t^{2\chi+1})$ approximation requires the application of $2\cdot5^{\chi-1}(L-1)+1$ individual oracles $\mathbb{O}_k(t)$, where $L$ is again the number of Hamiltonian terms.  More generally, we would represent them by oracle calls and we will refer to the number of exponentials $N_{\rm exp}$.
Using repetitions as in Eq.~\eqref{eq:reps}, the number of exponentials of the algorithm is of $\mathcal{O}(r L 5^{\chi-1})$ to approximate the actual time evolution $U(T)$ up to an accuracy of $ \varepsilon \in \mathcal{O}(r(T/r)^{2\chi+1})$ in operator distance. 

To enhance the comparability of product formulas with more involved multi-product formulas whose bounds have not been improved to the same extent, we use the following simple tail bound and note that better bounds exist, e.g. via commutator relations \cite{childsTheoryTrotterError2020}:
\begin{equation}
\label{eq:trotter_bound}
    \Norm{U(T)-S_{2\chi}(T/r)^r}\leq 2r\frac{\left(g_\chi \Lambda T/r\right)^{2\chi+1}}{(2\chi+1)!}
\end{equation}
Here, $g_{\chi}= (5/3)^{\chi-1} 4\chi/3$ originates in the exponential tail of the Trotter-Suzuki terms \cite{childsHamiltonianSimulationUsing} and $\Lambda = \sum_{k=1}^L\Norm{h_k}$.

The exponential dependence of the circuit depth on $\chi$ can pose a challenge for real-world implementations, and 
especially so within the NISQ regime. Altogether, the number of oracle calls needed to simulate the evolution up to an error $\varepsilon \leq 1$ is upper bounded by
\begin{equation}
    \label{eq:trotter_oracles}
    N_{\rm oracle} \leq \frac{2L5^{\chi - 1} \left(L\Lambda t\right)^{1+1/(2\chi)}}{\varepsilon^{1/(2\chi)}},
\end{equation}
as proven in Ref.~\cite{berry2007efficient}. An alternative complexity bound using commutator relations of the individual Hamiltonian terms can be found in Ref.~\cite{childsTheoryTrotterError2020}.

A known technique to decrease the number of required short time evolutions, i.e., oracle calls, is the use of multi-product formulas \cite{blanesExtrapolationSymplecticIntegrators1999,chinMultiproductSplittingRungeKuttaNystrom2010}. 
While the Trotter-Suzuki approximation cancels erroneous contributions of higher-order terms by adding backward evolutions,  multi-product formulas achieve the same cancellations by superposing different product formulas. 
Conventionally,  one employs multi-product formulas that describe the same time evolution (up to a fixed order of $t$), but whose erroneous higher-order contributions are of different strength and can thus be made to cancel. This is achieved by approximating the time evolution to the same Trotter-Suzuki order, but considering product formulas different in the number of time slices \cite{blanesExtrapolationSymplecticIntegrators1999,yoshidaConstructionHigherOrder1990}.
The Childs and Wiebe multi-product formula discussed in Ref.~\cite{childsHamiltonianSimulationUsing} is
of the form
\begin{equation}
    \label{eq:multi_product_formula}
    M_{K,2\chi}(t) \coloneqq \sum_{q=1}^{K+1}C_qS_{2\chi}(t/\ell_q)^{\ell_q},
\end{equation}
where $K$ is an integer defining a cutoff and $\lbrace \ell_q \rbrace$ is a set of pairwise different integers.
The coefficients $\lbrace C_q \rbrace$ are determined via
\begin{equation}
    \label{eq:LCU_vandermonde}
   \left( \begin{matrix}
    1 & 1 & \,  & \cdots & \, & 1 \\ 
    \ell_1^{-2\chi} & \ell_2^{-2\chi} && \cdots && \ell_{K+1}^{-2\chi} \\ 
    \ell_1^{-2\chi-2} & \ell_2^{-2\chi-2} && \cdots && \ell_{K+1}^{-2\chi-2} \\ 
     \vdots & \vdots && \ddots && \vdots \\ 
    \ell_1^{-2(K+\chi-1)} & \ell_2^{-2(K+\chi-1)} && \cdots && \ell_{K+1}^{-2(K+\chi-1)}
    \end{matrix} \right) \left( \begin{matrix} C_1 \\ 
    C_2\\ 
    C_3 \\ 
    C_4 \\ 
    \vdots \\ 
    C_{K+1} \\
    \end{matrix} \right) = \left( \begin{matrix} 1 \\ 
    0\\ 
    0 \\  
    0 \\ 
    \vdots \\ 
    0 
    \end{matrix} \right),
\end{equation}
ensuring the error terms in the multi-product formula vanish up to $\mathcal{O}(t^{2(K+\chi)})$, resulting in
\begin{align}
    \label{eq:childs-wiebe_bound}
    \Norm{U(t)-M_{K,2\chi}(t)}\leq&\left(1+g_\chi^{2(\chi+ K)+1}\sum\limits_{q=1}^{K+1}\left|C_q\right|\right)\frac{\left(\Lambda t \right)^{2(\chi+ K)+1}}{(2(\chi+K)+1)!}\\=&\mathcal{O}\left((\Lambda t)^{2(K+\chi)+1}\right),\nonumber
\end{align}
where we have again stated a simple tail bound for comparability and note that it could also exhibit a commutator scaling as discussed in Refs.~\cite{childsTheoryTrotterError2020,lowWellconditionedMultiproductHamiltonian2019}.

Unlike Childs and Wiebe, we will henceforth use the simplest version achieved by setting $\ell_q=q$ for all $q$ since this choice is the most favorable for a randomized implementation, which will become clearer in Section~\ref{sec:randomized_sampling_main_res}.

In Definitions~\ref{def:our_lcu_v1} and \ref{def:our_lcu_v2} and Section~\ref{sec:proofs}, we employ a different approach for  multi-product formulas and develop two techniques whose errors scale with $\mathcal{O}(t^{2\chi R+1})$, and where $R$ is comparable to $K+1$.

Multi-product formulas were firstly used for quantum simulation by Childs and Wiebe \cite{childsHamiltonianSimulationUsing}, who developed the \emph{linear-combinations-of-unitaries (LCU)} approach to directly implement multi-product formulas on a quantum system.
Note that sums of unitaries are not inherently physical operations since the unitary group is not closed under addition. Childs and Wiebe use a non-deterministic approach to implement multi-product formulas that can lead to large algorithmic depth. Berry et al.~\cite{berry2015hamiltonian} have implemented  LCU for a truncated Taylor series nearly perfectly but their use of oblivious amplitude amplification requires an additional register and a complex state preparation procedure.
Recently, Ref.~\cite{lowWellconditionedMultiproductHamiltonian2019} has improved the condition number of multi-product formulas and thus extended the use of LCU by amplitude amplification. However, these improvements have no impact on their asymptotic performance and are less favorable for randomized implementations.

Using \emph{randomized sampling} such as proposed in Algorithm~\ref{algo:random_sampling}, we can circumvent the need for these potentially deep circuits. Randomized compiling of multi-product formulas entails sampling among the individual terms of Eq.~\eqref{eq:multi_product_formula} requiring the implementation of just twice the number of oracle calls, rather than the entire linear combination and the overhead for its implementation.

\section{Randomized sampling using multi-product formulas}
\label{sec:randomized_sampling_main_res}
LCU-type algorithms, despite having a finite failure probability, are deterministic when successfully applied. In contrast to such deterministic implementations, there is a recent interest in randomized algorithms \cite{campbellRandomCompilerFast2019a,childsFasterQuantumSimulation2019,ouyangCompilationStochasticHamiltonian2020,chenQuantumSimulationRandomized2020}.
These novel results improve the performance of product formulas by randomizing the order of the short time evolution and introducing the alternative setting of randomized compiling.
With the goal of reducing the circuit depth required for the implementation of multi-product formulas, we build upon randomized compiling and propose the setting of randomized sampling: Given an observable $O$ and a quantum system in the initial state $\rho$, our goal is to find the expectation value of the observable after a time evolution $U=\ee^{-\ii Ht}$ governed by the Hamiltonian $H=\sum_{k=1}^L h_k$. That is, we wish to compute
\begin{equation}
    \label{eq:goal}
    \langle O(t)\rangle = {\rm tr}(O U\rho U^\dagger) \, .
\end{equation}
In the following, we describe a randomized algorithm for such a task, give convergence guarantees, and present novel multi-product formulas suited for this framework. Note that we will require one part of our system to be encoded as a single qubit, while the rest acts as the simulator.

\subsection{The randomized sampling framework}
\label{sec:rand_samp_algo}

We begin by proposing the randomized sampling framework summarized in Figure \ref{fig:sampling_algo}. After rescaling the coefficients of a given multi-product formula, we can apply the following simple algorithm to estimate Eq.~\eqref{eq:goal} in a randomized fashion based upon importance sampling:

\begin{algorithm}[Randomized sampling]
    \label{algo:random_sampling}
    Given an observable $O$ and an ensemble $\mV=\set{(V_k,p_k)}_{k=1}^M$ of $M$ unitaries $\lbrace V_k \rbrace$ and corresponding probabilities $\lbrace p_k \rbrace$, we consider $N$ independent repetitions of the following protocol:
    \begin{enumerate}
        \item Prepare $\ket{+}\!\!\bra{+}\otimes \rho \, $.
        \item Sample $V_\circ\stackrel{i.i.d.}{\sim}\mV$ and apply the anti-controlled unitary $\ket{0}\!\!\bra{0}\otimes V_\circ +\ket{1}\!\!\bra{1}\otimes \id \, $.
        \item Sample $V_\bullet\stackrel{i.i.d.}{\sim}\mV$ and apply the controlled unitary $\ket{0}\!\!\bra{0}\otimes\id +\ket{1}\!\!\bra{1}\otimes V_\bullet \,$.
        \item Perform a single shot of the POVM measurement associated with the observable $X\otimes O \, $.
        \item Store the measurement outcomes $\{o_j\} $.
    \end{enumerate}
\end{algorithm}
\begin{figure}[t!]
    \centering
    \begin{adjustbox}{max width=\linewidth}
    \input{figures/sampling_algo.tex}
    \end{adjustbox}
    \caption{Overview of the randomized sampling framework for estimating the dynamics of observables. After preparing a multi-product formula for importance sampling, we can run Algorithm~\ref{algo:random_sampling} to implement it in a simple, randomized fashion. While this overview already hints toward its use for time evolution and estimating the dynamics of observables, these results apply to general multi-product formulas.}
    \label{fig:sampling_algo}
\end{figure}
Note that it is in general sufficient to measure $X$ and $O$ separately and multiply the outcomes, rather than performing a joint measurement. This could be useful e.g. when the measurement of $O$ is available in a black-box fashion. As we show in more detail in Section~\ref{sec:proofs}, we find the following result for infinitely many runs of the algorithm:

\begin{corollary}[Sample mean convergence]
    \label{coro:sampling_expectation}
    Let $O$ be an observable, $\rho$ be an initial state and let $\mV = \set{(V_k,p_k)}_{k=1}^M$ be an ensemble of unitaries $\set{V_i}$ and corresponding probabilities $\set{p_i}$, such that 
    \begin{equation}
        \label{eq:lcu1}
        \EE_\mV[V_k] = \sum_{k=1}^Mp_kV_k = V \, .
    \end{equation}
    Then, the sample mean of Algorithm~\ref{algo:random_sampling} converges to the expectation value of the random measurement outcomes $\lbrace o_j \rbrace_j$, given by
    \begin{equation}
        \lim_{N \to \infty} \frac{1}{N}\sum_{j=1}^No_j = \trace{OV\rho V^\dagger} \, .
    \end{equation}
\end{corollary}
Furthermore, since an infinite number of measurements is infeasible, we give an expression for the number $N$ of repetitions of Algorithm~\ref{algo:random_sampling} sufficient for achieving the desired accuracy and confidence for the goal of randomized sampling. Using Hoeffding's inequality \cite{hoeffdingProbabilityInequalitiesSums1963}, we can formulate the following result:

\begin{theorem}[Randomized implementation of sums of unitaries]
\label{theo:randomized_implementation_unitaries}
Let $O$ be an observable, $\rho$ an initial state and let $\mV = \set{(V_k,p_k)}_{k=1}^M$ be an ensemble of unitaries $\set{V_k}$ and corresponding probabilities $\set{p_k}$, such that 
    \begin{equation}
        \label{eq:lcu2}
        \EE_\mV[V_k] = \sum_{k=1}^Mp_kV_k = V \, .
    \end{equation}
Then, for a fixed accuracy $\varepsilon\in(0,1)$ and confidence $\delta\in(0,1)$, a total of
\begin{equation}
       N\geq \frac{2\lVert O\rVert^2 \ln(2/\delta)}{\epsilon^2}
\end{equation}
repetitions of Algorithm~\ref{algo:random_sampling} suffice to accurately approximate the expectation value $\trace{OV\rho V^\dagger}$ using the sampling mean estimator, i.e.,
\begin{equation}
    \Abs{\frac{1}{N}\sum_{j=1}^No_j-\trace{OV\rho V^\dagger}}\leq \varepsilon
\end{equation}
with probability at least $1-\delta$.
\end{theorem}
So far, this framework is fairly general: Given a set $\mathcal{V}$, the algorithm samples from the linear transform with the unitary $V$. For quantum simulation, $V$ (a decent approximation to $U$) and $\mathcal{V}$ still have to be found.  This is where multi-product formulas come into play.  Multi-product formulas approximate $U$ by a weighted sum of product formulas, and so the featured product formulas and their weights could make up the set $\mathcal{V}$. The problem is that we require $\lbrace p_k \rbrace$ to form a probability distribution, i{.} e{.}, $p_k>0$ and $\sum_{k=1}^M p_k=1$. This disqualifies us from identifying the sets $\lbrace p_k \rbrace$ with $\lbrace C_q \rbrace$ straightforwardly.
While $\sum_{q}C_q =1$ holds for any multi-product formula per construction (see the first component of Eq.~\eqref{eq:LCU_vandermonde}), the sign of some $C_q$ will be negative \cite{shengSolvingLinearPartial1989}. We thus have to absorb the signs of these negative $C_q$  into their unitaries $V_q$, but then we find that the operation is not normalized as $\sum_{q}\abs{C_q}>1$. At this point we introduce the quantity $\Xi \coloneqq \sum_q |C_q| $, that we call \textit{resolution factor}. The resolution is practically used to define a probability distribution via $p_k = \left|C_k\right|/\Xi$ such that  Algorithm~\ref{algo:random_sampling} samples from
\begin{equation}
\label{eq:measurement_resolution}
\frac{1}{\Xi^2}\mathrm{tr}( O U\rho U^\dagger) \, .
\end{equation} 
As the time evolution is now approximated by $\Xi V$ rather than $V$, $\Xi$ factors into the number of required circuit evaluations. With $\Xi>1$ requiring to run the algorithm more often to achieve an error comparable with the situation where $\Xi=1$, this resolution factor can be regarded as a penalty.
This penalty is further amplified if a direct time evolution by time $t$ is insufficient and a repeated time evolution of short time slices $t/r<1$ is required. 
Since each of such repetitions will result in additional factors of $\Xi$ leading to an exponential scaling of the resolution factor in the number of repetitions, this behavior prohibits the use of a large number of repetitions as is usual for product formulas.
On the other hand, this randomized implementation does not rely on post-processing and the increase of some success probability via well-conditioning as is the case for the deterministic, block-encoding implementation, which in turn would increase the number of required circuit evaluations drastically.
However, especially in the regime of early quantum computers with capabilities between those of NISQ devices and large fault-tolerant quantum computers, where one would use few repetitions, a deeper circuit may pose a bottleneck much tighter than a slightly higher number of circuit evaluations would. We, therefore, envision these results to be of particular interest in this intermediate regime, where one might be interested in simulating the dynamics of quantum systems instead of performing full-fledged quantum phase estimation.

We can make the following guarantees for approximating time evolutions in randomized sampling with a nontrivial resolution  and a finite number of measurements:  

\begin{theorem}[Approximating unitaries with a finite number of measurements,]
\label{theo:time_evolution}
    Let $U$ be a unitary, $O$ an observable, $\rho$ some initial state and let $\mV = \set{(V_k,p_k)}_{k=1}^M$ be an ensemble of unitaries $\set{V_k}$ with a corresponding probability distribution $\set{p_k}$ such that
    $\EE_\mV[V_k] =\sum_{k=1}^Mp_kV_k= V$.  Let there be a constant factor $\Xi \in \mathbb{R}_+$ such that 
    \begin{equation}
        \label{eq:approx_time_evolution}
        \norm{\Xi V-U}\leq \varepsilon \, .
    \end{equation}
    Then, it is sufficient to run
    \begin{equation}
    N\geq  \frac{2 \norm{O}^2\ln \left(2/\delta\right) \Xi^4}{\varepsilon^2}
    \end{equation}
    repetitions of Algorithm~\ref{algo:random_sampling} to achieve
    \begin{equation}
        \Abs{\frac{\Xi^2}{N}\sum_{j=1}^No_j-\mathrm{tr}(OU\rho U^\dagger)}\leq \left(1+3\norm{O}\right)\varepsilon\, ,
    \end{equation}
    with probability at least $1-\delta$.
\end{theorem}

\subsection{Custom multi-product formulas}
\label{ssec:custom_mpfs}
The multi-product formulas of Childs and Wiebe can easily be adapted to this randomized sampling scheme. Since previous improvements of multi-product formulas focused on increasing success probabilities by changing their condition number \cite{lowWellconditionedMultiproductHamiltonian2019}, which has no impact on their asymptotic error scaling, we can ignore constraints encountered in LCU techniques. Our sole interest lies in a low algorithmic depth and a moderate resolution, so we choose to set $\ell_q = q$ for all 
$q=1, \dots, K+1$ in Eq.~\eqref{eq:multi_product_formula} and \eqref{eq:LCU_vandermonde}. 
However, those multi-product formulas make relatively little use of the Trotter-Suzuki order $2\chi$ of its components. In fact, the only reason not to minimize the depth by using  $S_2(\cdot)$ blocks is the high resolution factor.
To this end, we present a family of novel multi-product formulas tailored toward the randomized sampling framework with improved scaling in the Trotter-Suzuki order that can be optimized for their resolution factor. A brief overview of the process leading to a suitable multi-product formula is presented in Figure~\ref{fig:prescription}
\begin{figure}[t!]
    \centering
    \begin{adjustbox}{max width=\linewidth}
    \input{figures/prescription_tikz.tex}
    \end{adjustbox}
    \caption{Overview of the steps necessary to arrive at a suitable multi-product formula as discussed in the following section. While the matching multi-product formula is superior to the closed-form version, the required solution to this system of nonlinear equations might not always be available}
    \label{fig:prescription}
\end{figure}
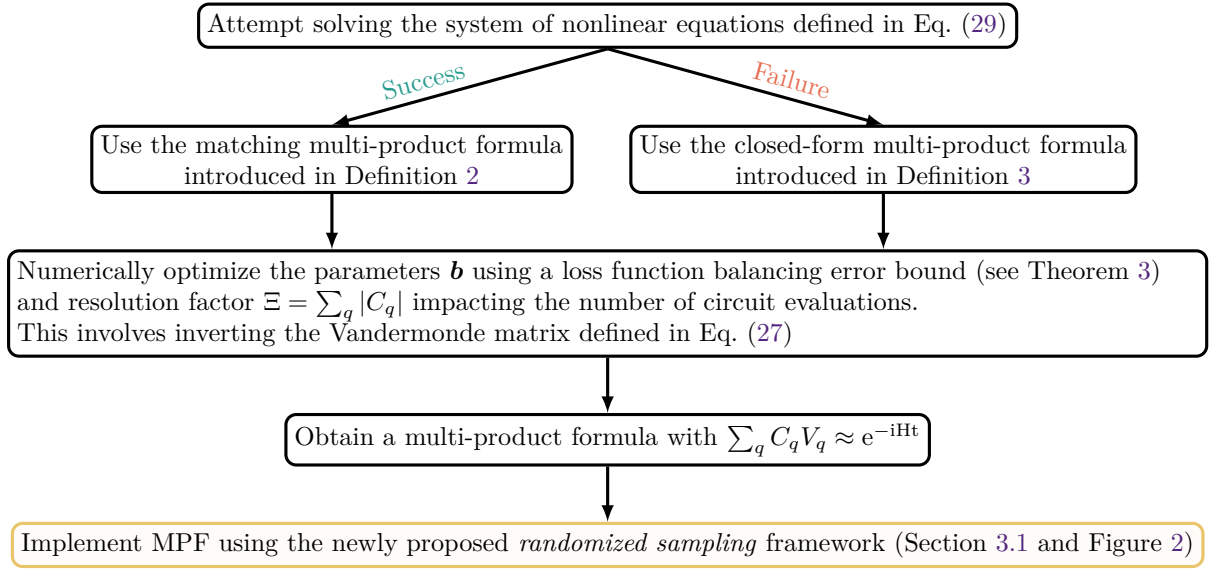

While Childs and Wiebe's approach manipulates higher-order error terms, we employ a construction to modulate the entire Taylor expansion. 

\begin{definition}[Linear combination of time evolution operators]
\label{def:new_multi_prod_base}
Let $\chi\geq1$ and $R\geq1$ be integers and $S_{2\chi}(t)$ a Trotter-Suzuki product formula approximation to $U(t)$ as defined in Eq.~\eqref{eq:trotter_suzuki} with $\Norm{S_{2\chi}(t)-U(t)}\in \mathcal{O}(t^{2\chi+1})$. Then, for any $t\in\mathbb{R}$ and arbitrary $\bb = (b_0,\ldots, b_{2\chi R})^\top\in \mathbb{R}^{2\chi R +1}$ and $\nn = (\nu_0,\ldots, \nu_{2\chi R})^\top \in \mathbb{R}^{2\chi R +1}$ , we define the multi-product formula $\mL_{2\chi,R}(\nn,\bb, t)$ as
\begin{equation}
    \label{eq:L_definition}
    \mL_{2\chi,R}(\nn,\bb, t) \coloneqq \sum_{q=0}^{2\chi R}C_q(\nn, \bb) \, S_{2\chi}\!\left(b_qt\right),
\end{equation}
where
\begin{equation}
    \label{eq:vandermonde}
   \left( \begin{matrix}
    1 & 1 & \,  & \cdots & \, & 1 \\
    b_0 & b_1 && \cdots && b_{2\chi R} \\
    b_0^2 & b_1^2 && \cdots && b_{2\chi R}^2 \\ 
    b_0^3 & b_1^3 && \cdots && b_{2\chi R}^3 \\
     \vdots & \vdots && \ddots && \vdots \\
    \, b_0^{2\chi R} & \, b_1^{2\chi R} && \cdots && \, b_{2\chi R}^{2\chi R} 
    \end{matrix} \right) \left( \begin{matrix} C_0 \\
    C_1\\
    C_2 \\ 
    C_3 \\
    \vdots \\
    C_{2\chi R} 
    \end{matrix} \right) =  \left( \begin{matrix} \nu_0 \\ \nu_1 \\ \nu_2 \\ \nu_3 \\ \vdots \\ \nu_{2\chi R} \end{matrix}\right)
\end{equation}
for some  $(C_0, \, C_1, \, \dots \, , \, C_{2\chi R})^\top \in \mathbb{R}^{2\chi R +1}$.
\end{definition}
Here, the coefficients $\lbrace C_q \rbrace$ are related to the parameters $\bb$ and $\nn$ via the inverse of the corresponding Vandermonde matrix in Eq.~\eqref{eq:vandermonde} as is discussed in detail in Section~\ref{sec:proofs}. The parameters $\nn$ will have an important role in the following construction, whereas $\bb$ can be tuned to strike a balance between the resolution factor $ \Xi $, error bound $\varepsilon$, and condition number. While previous approaches focused on optimizing $\bb$ for the condition number relevant for the success probability when using block-encodings and amplitude amplification, we instead numerically optimize them for the resolution factor and obtain improved error bounds along the way. Note that a similar construction has recently been used to approximate energy measurements in the fully analog setting \cite{bespalova2020hamiltonian}. We can now turn toward the definition of the first new multi-product formula. 

\begin{definition}[Matching multi-product formula]
\label{def:our_lcu_v1}
Let $\chi\geq1$ and $R\geq1$ be integers and $\mL_{2\chi,R}(\nn,\bb, t)$ as defined in Definition~\ref{def:new_multi_prod_base}. 
Then, for any $t\in\mathbb{R}$, we define the multi-product formula $\widetilde{M}_{2\chi,R}^{(\mathrm{m})}(t)$ as
\begin{equation}
    \widetilde{M}_{2\chi,R}^{(\mathrm{m})}(t)\coloneqq \prod_{r=1}^R\mL_{2\chi,R}\left(\nn^{(r)},\bb^{(r)},t\right)
\end{equation}
with $\nu_k^{(r)} = 0$ for $k> 2\chi$ and all $ r = 1 \dots R$. The remaining $\left\lbrace \nu_q^{(r)} \right\rbrace$ are fixed by
\begin{equation}
\label{eq:constraint}
\sum_{k_1+k_2+\ldots+k_R =k} \frac{\nu_{k_1}^{(1)}\nu_{k_2}^{(2)}\cdots \nu_{k_R}^{(R)} }{k_1!\,k_2! \, \cdots \, k_R!}=\frac{1}{k!}
\end{equation}
for all $0 \leq k \leq 2\chi R$, whereas the parameters $\lbrace \boldsymbol{b}^{(r)} \rbrace$ can be chosen arbitrarily.
\end{definition}
Multiplying the coefficients $C_q$ of the individual building blocks $\mL_{2\chi,R}(\nn^{(r)} ,\bb^{(r)}, t)$ of the matching multi-product formula results in a resolution factor of
\begin{equation}
\Xi^{(\mathrm{m})}\coloneqq\prod_r\left(\sum_q\Abs{C_q\!\left(\nn^{(r)}, \bb^{(r)}\!\right)}\right).
\end{equation}
For the time evolution to work, the set of vectors $\left\lbrace \nn^{(r)} \right\rbrace$ has to be found satisfying the constraint in Eq.~\eqref{eq:constraint} for different $R$ and $2\chi$. It is important to note that obtaining the coefficients $\nn$ for the matching multi-product formula requires solving the system of nonlinear equations defined by Eq.~\eqref{eq:constraint}, whose complexity scales at least exponentially in $\chi R$. In this work, we have used numerical solvers which converge in a few seconds for $\chi R\leq 10$, which we view as reasonably high as they provide an approximation up to $\mathcal{O}(t^{21})$, which we view as suitable for early quantum computers.

When solving the system of nonlinear equations becomes unfeasible, be it due to slow convergence or the desire for significantly better approximations, we can employ a second version of the multi-product formula in which the $\left\lbrace \nn^{(r)} \right\rbrace$ are already determined at the cost of a slightly larger resolution factor $\Xi$:
\begin{definition}[Closed-form multi-product formula]
\label{def:our_lcu_v2}
Let $\chi\geq1$ and $R\geq1$ be integers and $\mL_{2\chi,R}(\nn,\bb, t)$ as defined in Definition~\ref{def:new_multi_prod_base}. Then for any $t\in\mathbb{R}$, we define the multi-product formula $\widetilde{M}_{2\chi,R}^{(\mathrm{cf})}(t)$ as
\begin{equation}
    \widetilde{M}_{2\chi,R}^{(\mathrm{cf})}(t)\coloneqq \sum_{r=1}^R\left(\mL_{2\chi,R}\!\left(\nn^{(0)},\bb^{(0)}, t\right)\right)^{r-1}\mL_{2\chi,R}\!\left(\nn^{(r)},\bb^{(r)}, t\right)
\end{equation}
with
\begin{align}
    \nu^{(0)}_k &= \begin{cases}
    1, & \text{for} \; k =  2\chi\\
    0, & \text{else},
    \end{cases}\\
    \nu^{(1)}_k &= \begin{cases}
    1, & \text{for} \; k\leq 2\chi\\
    0, & \text{else},
    \end{cases}\\
    \nu^{(n)}_k &= \begin{cases}
    \frac{k!\left((2\chi)!\right)^{n-1}}{(2\chi (n-1) + k)!}, & \text{for}\; 0 < k\leq 2\chi\\
    0, & \text{else},
    \end{cases}
\end{align}
for all $\; 0\leq k \leq 2\chi R \;$ and $\; 1 < n \leq R $. The parameters $\lbrace \boldsymbol{b}^{(r)} \rbrace$ can be chosen arbitrarily.
\end{definition}
Multiplying the coefficients of the individual $\mL_{2\chi,R}(\nn^{(r)} ,\bb^{(r)}, t)$, results in a resolution factor for the closed-form multi-product formula of
\begin{equation}
    \Xi^{(\mathrm{cf})}:=\sum_r \left(\sum_q \Abs{C_q\!\left(\nn^{(0)},\bb^{(0)}\!\right)}\right)^{r-1}\left(\sum_q \Abs{C_q\!\left(\nn^{(r)},\bb^{(r)}\!\right)}\right).
\end{equation}

For these multi-product formulas, we can prove the following two results:
\begin{corollary}[Low depth unitary approximation]
\label{cor:error_scaling}
Let $\widetilde{M}_{2\chi,R}(t)$ be a multi-product formula constructed according to Definition~\ref{def:our_lcu_v1} or \ref{def:our_lcu_v2}. We then find that
\begin{equation}
    \Norm{U(t)-\widetilde{M}_{2\chi,R}(t)} =  \mathcal{O}\left(t^{2\chi R+1}\right).
\end{equation}
\end{corollary}

\begin{theorem}[Error bound for custom multi-product formulas]
\label{theo:error_bound}
For $\widetilde{M}_{2\chi,R}(t)$ being a multi-product formula constructed according to Definitions~\ref{def:our_lcu_v1} or \ref{def:our_lcu_v2}, we have
\begin{equation}
    \Norm{U(t)-\widetilde{M}_{2\chi,R}(t)} \leq \left(1 + \zeta g_{\chi}^{2\chi R+1}\right) \frac{\left( \Lambda t \right)^{2\chi R+1}}{\left(2\chi R+1\right)!} \,, 
\end{equation}
with $g_{\chi}\ \coloneqq \frac{4\chi}{3}\left(\frac{5}{3}\right)^{\chi-1}$, $\Lambda \coloneqq \sum\limits_k\Norm{h_k}$ and corresponding
\begin{align}
     \zeta^{(\mathrm{m})} \coloneqq& \sum_{q_1=0}^{2\chi R}\sum_{q_2=0}^{2\chi R}\cdots \sum_{q_r=0}^{2\chi R} \Abs{C_{q_1}^{(1)}C_{q_2}^{(2)} \cdots  C_{q_r}^{(R)}}\left( \abs{b_{q_1}^{(1)}}+\ldots +\abs{b_{q_r}^{(R)}}\right)^{2\chi R+1}\\
     \leq& \Xi^{(\mathrm{m})}\left(R\max_{q,r}\abs{b_q^{(r)}}\right)^{2\chi R}\\
    \zeta^{(\mathrm{cf})}\coloneqq& \sum_{r=1}^R\sum_{q_1=0}^{2\chi R}\sum_{q_2=0}^{2\chi R}\cdots \sum_{q_r=0}^{2\chi R} \Abs{C_{q_1}^{(0)}C_{q_2}^{(0)} \cdots  C_{q_{r-1}}^{(0)}C_{q_{r}}^{(r)}}\left( \abs{b_{q_1}^{(0)}}+\ldots+\abs{b_{q_{r-1}}^{(0)}} +\abs{b_{q_r}^{(r)}}\right)^{2\chi R+1}\\
    \leq& \Xi^{(\mathrm{cf})}\left(R\max_{q,r}\abs{b_q^{(r)}}\right)^{2\chi R},
\end{align}
for matching and closed-form multi-product formulas respectively.
\end{theorem}

Consequently, these formulas can be used for randomized sampling via Theorem~\ref{theo:time_evolution}.
A numerical comparison of the error bounds of all presented formulas can be found in Fig.~\ref{fig:compare_bounds}, while an overview of all relevant quantities is provided in Table~\ref{tab:compare_methods}.

While the formula dependent $\zeta^{(\mathrm{m})}_{{\chi,R}}$ and $\zeta^{(\mathrm{cf})}_{{\chi,R}}$ themselves are not too illuminating, they can be upper bounded by the corresponding resolution factors multiplied by a term depending on the magnitudes of the used parameters $\bb$. However, these upper bounds are not tight which is why their exact versions are used in any explicit calculations. Nevertheless, they showcase that a bad choice of $\bb$ can lead to significant penalties in the performance of the resulting multi-product formula, further stressing the need for optimizing with respect to these parameters.

The matching and closed-form multi-product formulas both rely on products of $\mathcal{L}_{2\chi, R}(\boldsymbol{\nu}^{(r)}, \boldsymbol{b}^{(r)}, t)$. For Algorithm ~\ref{algo:random_sampling}, these products can either be expanded and the resulting operators and coefficients be identified with the sets $\lbrace V_k \rbrace$, $\lbrace p_k \rbrace$, or sets of $(V_\circ^{(1)},V_\bullet^{(1)}), \dots, (V_\circ^{(R)},V_\bullet^{(R)})$ be drawn and applied  as in Fig.~\ref{fig:circuit}. For the matching multi-product formula,  all  $V_\circ^{(r)}$ and  $V_\bullet^{(r)}$ are drawn independently from each other while for the closed-form multi-product formula, the set from which $V_{\circ / \bullet}^{(r+1)}$ is drawn depends on the set from which  $V_{\circ / \bullet}^{(r)}$ has been sampled.
\begin{figure}
    \centering
    \includegraphics{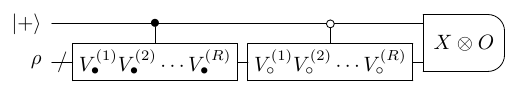}       
    \caption{Quantum circuit for randomized sampling with a multi-product formula $\prod_{r=1}^R\left(\sum_k p^{(r)}_k V^{(r)}_k\right)$. The unitaries $V_{\circ}^{(r)}, V_{\bullet}^{(r)}$ are drawn independently at random from $\lbrace V^{(r)}_k \rbrace_k$  according to the probabilities $\lbrace p^{(r)}_k \rbrace_k$ for all $r=1 \dots R$, such that the circuit samples from $\EE(\langle O \rangle)$. Again, it is in general sufficient to multiply the measurement outcomes instead of performing a joint measurement.}
    \label{fig:circuit}
\end{figure}

Furthermore, since the above framework randomizes only over product formulas, we are not necessarily bound to Trotter-Suzuki formulas and one could also think of a doubly-stochastic version in which the short-term evolutions comprising the product formulas are sampled randomly as well. In that sense, it is  allowed to construct the multi-product formulas with building blocks of 
\begin{align}
    \widehat{S}_{2}(t)= \frac{1}{2}\left( S_1^{\phantom{\dagger}}(t) + S_1^{\dagger}(-t) \right) \, ,
\end{align}
rather than \eqref{eq:trotter_suzuki_0}. This is not possible for Childs and Wiebe's multi-product formulas, which require symmetric building blocks and are thus limited to Trotter-Suzuki formulas.

\section{Proofs}
\label{sec:proofs}
In this section, we provide proofs of the statements presented above. In Section~\ref{subsec:rando_sampl}, we are proving our statements regarding the errors and uncertainty of the randomized sampling procedure. We then turn our attention to the matching and closed-form multi-product formula, providing the intuition behind their construction and verifying their error scaling in Section~\ref{subsec:mpf_1}. We finally analyze the upper bounds for their error (which is the error of the average time evolution they describe) in Section~\ref{subsec:errors}.
\subsection{Randomized sampling}
\label{subsec:rando_sampl}

Let us start by following Algorithm~\ref{algo:random_sampling} step by step. The controlled and anti-controlled applications of $V_\circ$ and $V_\bullet$ are defined as
\begin{align}
    \mathrm{\overline{C}}V_\circ &\coloneqq \ket{0}\!\!\bra{0}\otimes V_\circ +\ket{1}\!\!\bra{1}\otimes \id ,\\
    \mathrm{C}V_\bullet&\coloneqq\ket{0}\!\!\bra{0}\otimes\id +\ket{1}\!\!\bra{1}\otimes V_\bullet \, ,
\end{align}
and the initial state $\rho$ will be transformed to
\begin{align}
     \widetilde{\rho} &= \mathrm{C}V_\bullet\mathrm{\overline{C}}V_\circ \, \left( \left| +\right\rangle\!\!\left\langle + \right|\otimes \rho\right) \,(\mathrm{\overline{C}}V_\circ)^\dagger(\mathrm{C}V_\bullet)^\dagger \\
     &= \frac{1}{2}\left(\ket{0}\!\!\bra{0}\otimes V^{\phantom{\dagger}}_\circ\rho V_\circ^\dagger+\ket{0}\!\!\bra{1}\otimes V^{\phantom{\dagger}}_\circ\rho V_\bullet^\dagger+\ket{1}\!\!\bra{0}\otimes V^{\phantom{\dagger}}_\bullet\rho V_\circ^\dagger+\ket{1}\!\!\bra{1}\otimes V^{\phantom{\dagger}}_\bullet\rho V_\bullet^\dagger\right)
\end{align}
after the third step of the protocol. The expectation value for the succeeding measurement of $X\otimes \id$ is then given by
\begin{equation}
    \label{eq:quantum_average}
    \trace{X\otimes O \widetilde{\rho}} = \frac{1}{2}\trace{O(V^{\phantom{\dagger}}_\circ\rho V_\bullet^\dagger +V^{\phantom{\dagger}}_\bullet\rho V_\circ^\dagger)}.
\end{equation}
Here, we note that is in general not necessary to perform a joint measurement but sufficient to multiply the measurement results. This might be beneficial for cases in which the measurement of $O$ is only available in a black-box fashion.
We now use that $V_\circ$ and $V_\bullet$ are sampled independently from the same ensemble
\begin{equation}
V_\circ, V_\bullet\stackrel{i.i.d.}{\sim} \mV=\set{(V_k,p_k)}_{k=1}^M
\end{equation}
such that
\begin{equation}
    \label{eq:classical_average}
    \EE_\mV[V_\circ]=\EE_\mV[V_\bullet]=\sum_{k=1}^Mp_kV_k= V \, .
\end{equation}
Then, we can combine the quantum average form Eq.~\eqref{eq:quantum_average} with the classical average from Eq.~\eqref{eq:classical_average} to conclude that the expectation value of the random single-shot measurement outcome $o$ is given by
\begin{align}
    \EE[o]&=\EE_\mV\left[\trace{X\otimes O \widetilde{\rho}}\right]
    = \trace{O V\rho V^\dagger},
\end{align}
which proves Corollary~\ref{coro:sampling_expectation}.

For real-world applications, we will never achieve a perfect expectation value. It is therefore vital to inspect the single-shot behavior of Algorithm~\ref{algo:random_sampling} and give an estimate for the number of iterations required to achieve the desired precision.
Since the possible outcomes of measuring $X$ on the auxiliary qubit are $\pm1$ while the absolute of the outcomes of the $O$ measurement are bounded by $\norm{O}$, each single-shot measurement outcome $o_j\in[-\norm{O},\norm{O}]$
Consequently, by applying Hoeffding's inequality, we can prove the validity of Theorem~\ref{theo:randomized_implementation_unitaries}.

\begin{proof}[Proof of Theorem~\ref{theo:randomized_implementation_unitaries}]
Since the single-shot measurement outcomes $\lbrace o_j \rbrace$ obtained from $N$ iterations of Algorithm~\ref{algo:random_sampling} are individually  bounded by the interval $[-\norm{O},\norm{O}]$, Hoeffding's inequality states
\begin{equation}
\mathrm{Prob}\left(\Abs{\frac{1}{N}\sum_{j=1}^N o_j -  \trace{O V\rho V^\dagger}}\geq \varepsilon\right)\leq 2\exp{\left(-\frac{N\varepsilon^2}{2\norm{O}^2}\right)}.
\end{equation}
For a fixed accuracy $\varepsilon\in(0,1)$ and confidence $\delta\in(0,1)$, it is therefore sufficient to perform
\begin{equation}
    N\geq \frac{2\lVert O\rVert^2 \ln(2/\delta)}{\epsilon^2}
 \end{equation}
repetitions of Algorithm~\ref{algo:random_sampling} to achieve the desired accuracy and confidence, concluding the proof.
\end{proof}
Given recent developments, for example in Ref.~\cite{huangPredictingManyProperties2020a}, one might wonder whether
substantial improvements are possible by using a more refined estimator, most notably median-of-means.
Unfortunately, this is not very realistic in most scenarios. For observables that obey $O^2 = \mathbb{I}$, such as local and global Pauli observables, the variance of a single-shot outcome $o_j \in \left[-1,1\right]$ becomes $\mathrm{Var} \left[o_j \right] = 1 - \mathbb{E} \left[ o_j \right] = \mathcal{O}(1)$. In this case, the variance is of the same order as the magnitude and it is impossible to (asymptotically) improve over Hoeffding's inequality (asymptotic normality) \cite{lecam1960asymptotic}. However, median-of-means could still be used in this setting to take advantage of additional information about the variance, in cases where it is available.
Now, we additionally assume that $V$ times the resolution $\Xi$ approximates the time evolution operator $U$, i.e.,
\begin{equation}
    \label{eq:U_and_V_close}
    \norm{\Xi V-U} \leq \varepsilon \, .
\end{equation}
Furthermore, we will use the following result.

\begin{lemma}[Closeness of expectation values]
\label{lem:operator_diamond}
Let $U$ be unitary and $V$ and $\Xi$ as given above such that Eq.~\eqref{eq:U_and_V_close} holds. Furthermore, fix a state $\rho$ and observable $O$. Then,
\begin{equation}
    \Abs{\trace{OU\rho U^\dagger}-\Xi^2\trace{O V\rho V^\dagger}}\leq3\varepsilon\norm{O} \, .
\end{equation}
\end{lemma}
\begin{proof}
We begin by defining $\widetilde{U}$ as the difference between the exact and approximated time evolution
\begin{equation}
    \Xi V = \Xi \sum_{k=1}^Mp_kV_k = U+\widetilde{U}.
\end{equation}
According to Eq.~\eqref{eq:U_and_V_close} we find that $\norm{\widetilde{U}}\leq\varepsilon$. Consequently,
\begin{align}
\Abs{\trace{OU\rho U^\dagger} - \Xi^2\trace{OV\rho V^\dagger}}\notag   &= \Abs{\trace{\widetilde{U}^\dagger O U\rho} + \trace{U^\dagger O \widetilde{U}\rho}+\trace{\widetilde{U}^\dagger O\widetilde{U}\rho}}\\
&\leq 2\norm{\widetilde{U}}\norm{U}\norm{O}+\norm{\widetilde{U}}^2\norm{O} \\ & \leq \left(2\varepsilon+\varepsilon^2\right)\norm{O}\\
&\leq 3\varepsilon\norm{O}.
\end{align}
\end{proof}

By combining the insights from the previous discussion, we can now tackle Theorem~\ref{theo:time_evolution}:
\begin{proof}[Proof of Theorem~\ref{theo:time_evolution}]
    Add and subtract $\Xi^2\trace{O V\rho V^\dagger}$ to find
    \begin{equation}
        \Abs{\frac{\Xi^2}{N}\sum_{j=1}^No_j-\trace{OU\rho U^\dagger}}\leq \Xi^2 \Abs{\frac{1}{N}\sum_{j=1}^No_j - \trace{O V\rho V^\dagger}}+ \Abs{\vphantom{\sum_{j=1}^N}\Xi^2\trace{O V\rho V^\dagger}-\trace{OU\rho U^\dagger}}.
    \end{equation}
    By applying Theorem~\ref{theo:randomized_implementation_unitaries} with accuracy $\varepsilon/\Xi^2$ to the first term and using Lemma~\ref{lem:operator_diamond}, we find that this upper bound is given by $(3\norm{O}+1) \varepsilon$.
\end{proof}

\subsection{Construction of the new multi-product formulas}
\label{subsec:mpf_1}
In Definitions~\ref{def:our_lcu_v1} and \ref{def:our_lcu_v2}, we propose two alternatives to Childs and Wiebe's multi-product formulas, which reduce the number of circuit evaluations required for a randomized implementation and have improved error scaling. In the following, we will motivate their construction and stress the advantages they provide.
First of all, note that every approximation of the exact time evolution $\widetilde{U}(t) \approx \ee^{-\ii tH}$  can be written as
a sum of operators $\hat{A}_{k}$ such that
\begin{align}
    \widetilde{U}(t) = \sum_{k=0}^{\infty} \hat{A}_{k} \, t^k \, ,
\end{align}
with $\hat{A}_{0} = \mathbb{I}$.
For any approximated time evolution with an error of $\mathcal{O}(t^{m+1})$, we find that $\hat{A}_k = (iH)^k/{k!}$ for all $k\leq m$. In other words, the Taylor expansion of the exact time evolution and its approximation has the same Taylor expansion for the first $m$ non-trivial terms.

For the standard Trotter-Suzuki formula with $\widetilde{U}(t) = S_{2\chi}(t)$, we have 
\begin{equation}
\hat{A}_k = (iH)^k/{k!} 
\end{equation}
for all $k\leq 2\chi$. For $k>2\chi$, these operators $\hat{A}_k$ resemble uncontrolled, erroneous operators.
In Definition~\ref{def:new_multi_prod_base}, we propose a superposition of approximations with differently scaled times $S_{2\chi}(b_nt)$, introducing controllable modulation parameters $\nu_k$ up to a Taylor order of $2\chi R$. Its Taylor expansion is given by
\begin{align}
    \label{eq:decomp0}
    \mL_{2\chi, R}(\nn,\bb, t) &= \sum_{q=0}^{2\chi R}C_q (\nn,\bb)\,  S_{2\chi}(b_q t) \\ &= \sum_{k=0}^{\infty} \left( \sum_{q=0}^{2\chi R} C^{\phantom{k}}_q\!(\nn,\bb) \, b_q^k\right) \hat{A}_k t^k \\ &= \sum_{k=0}^{2\chi R} \nu_k\,  \hat{A}_k t^k \; + \; \mathcal{O}\left(t^{2\chi R + 1}\right), \label{eq:decomp2}
\end{align}
where $C_q, b_q \in \mathbb{R}$ and $\nu_k = ( \sum_q C^{\phantom{k}}_q  b_q^k)$ is fixed via the linear transformation of the coefficients $$\boldsymbol{C} = (C_0, C_1,\dots , C_{2 \chi R} )^\top$$ with the Vandermonde matrix  $B_{j,k} = b_k^{j-1}$, $B\boldsymbol{C} =\nn$ as 
in Eq.~\eqref{eq:vandermonde}. The condition number of the Vandermonde matrix
as the product of its Hilbert-Schmidt norm and the norm of the pseudo-inverse can be
bounded from above and below by explicit expressions 
involving the vector defining the Vandermonde matrix \cite{Bazan}.
Choosing the vector $\boldsymbol{b}$, we can  calculate the coefficients $\boldsymbol{C}$ for a fixed solution vector $\nn$ using the inverse  Vandermonde matrix $B^{-1}$, which is found exactly to be \cite{el-mikkawyExplicitInverseGeneralized2003,Knuth}
\begin{align} 
     \left(B^{-1}\right)_{j,k} = \frac{(-1)^{k-1}}{\prod\limits_{m \, \in \,  \boldsymbol{\mu}(j)} \left(b_m - b_j \right)} \sum_{\substack{\boldsymbol{a} \in \mathbb{F}_2^{^{2\chi R}} \\ |\boldsymbol{a}| = 2\chi R-k}}   \prod_{i=0}^{2\chi R-1} \left(b_{\mu_i(j)}\right)^{a_i} \\
     \notag \text{with} \quad \boldsymbol{\mu}(j) = (0,1, \ldots, j-1, j+1,  \ldots  , 2\chi R ) \, ,
\end{align}
where the sum runs over all binary strings $\boldsymbol{a}=(a_0 \,a_1 \, \dots \, a_{2\chi R -1})$ of length $2\chi R$ and Hamming weight $2\chi R - k$. However, we have found that for numerical purposes, a matrix inversion of $B$ clearly outmatches the analytical computation of $B^{-1}$ in terms of runtime.

Using Eq.~\eqref{eq:decomp2}, we can now manipulate the Taylor expansion of a Trotter-Suzuki block $S_{2\chi}(t)$ by replacing it with some $\mathcal{L}_{2\chi, R}(\boldsymbol{\nu}, \boldsymbol{b}, t)$. The key insight here is that while the operators $\hat{A}_k$ are only correct for Taylor orders $k\leq 2\chi$, we gain control of the prefactors up to order $2\chi R$. Consequently, we can eliminate all erroneous $\hat{A}_k$ for $2\chi<k\leq 2\chi R$ by setting the corresponding $\nu_k$ to zero.
This allows us to construct the matching and closed-form multi-product formulas using blocks of $\mathcal{L}_{2\chi, R}(\boldsymbol{\nu}, \boldsymbol{b}, t)$ with different $\boldsymbol{\nu}$ (and possibly $\boldsymbol{b}$).

\subsubsection{Matching multi-product formula}
The first version of our proposed multi-product formula builds upon the multiplication of $R$  multi-product building blocks $\mL_{2\chi,R}(\nn^{(r)} ,\bb^{(r)}, t)$ that according to Eq.~\eqref{eq:decomp2} can be written as
\begin{equation}
    \mL_{2\chi,R}\left(\nn^{(r)},\bb^{(r)},t\right)=\sum_{k=0}^{2\chi}\frac{\nu^{(r)}_k}{k!} \left(-\ii Ht\right)^k + \sum_{k=2\chi+1}^{2\chi R} \nu_k^{(r)} \hat{A}_k t^k+
    \mathcal{O}\left(t^{2\chi R+1}\right) \, , \label{eq:scaling1}
\end{equation}
where we can eliminate the second sum by setting $\nu_k=0$ for $2\chi<k\leq2\chi R$.
Their product now yields
\begin{equation}
    \prod_{r=1}^R\mL_{2\chi,R}\left(\nn^{(r)},\bb^{(r)},t\right) = \sum_{k=0}^{2\chi R}\mu_k \left(-\ii Ht\right)^k + \mathcal{O}(t^{2\chi R+1})
\end{equation}
with
\begin{equation}
\mu_k \,= \sum_{i_1+i_2+\ldots+i_R =k} \frac{\nu_{i_1}^{(1)}\nu_{i_2}^{(2)}\cdots \nu_{i_R}^{(R)} }{i_1!i_2!\cdots i_R!} \, .
\end{equation}
To mimic the exact time evolution up to $\mathcal{O}(t^{2\chi R + 1})$, we require 
$\mu_k=1/k!$, leading to Definition~\ref{def:our_lcu_v1}.

\subsubsection{Closed-form multi-product formula}
The closed-form version of the presented multi-product formula  sums products of building blocks $\mL$.  To motivate the specific construction, it is useful to take a look at the specific case of $2\chi=4$ and $R=3$ for some choice of  $\bb$ and $t$ using the shorthand 
\begin{equation}
\mL\left( \begin{smallmatrix} \nu_0 \\ : \\ \nu_{2\chi R}  \end{smallmatrix}\right) := \mL(\nn, \bb, t),
\end{equation}
to get
\begin{equation}
\label{eq:LCUv2_idea}
    \widetilde{M}^{(\mathrm{cf})}_{4,3} = 
    \mL\left( \begin{matrix}
    1\\%0
    \boldsymbol{1}\\%1
    \boldsymbol{1}\\%2
    \boldsymbol{1}\\%3
    \boldsymbol{1}\\%4
    0\\%5
    0\\%6
    \vdots%7
    \end{matrix} \right) +  4!\, \mL\left( \begin{matrix}
    0\\%0
    \boldsymbol{0}\\%1
    \boldsymbol{0}\\%2
   \boldsymbol{0}\\%3
    \boldsymbol{1}\\%4
    0\\%5
    0\\%6
    \vdots%7
    \end{matrix} \right) \mL\left( \begin{matrix}
    0\\%0
    \boldsymbol{1!/5!}\\%1
    \boldsymbol{2!/6!}\\%2
   \boldsymbol{3!/7!}\\%3
    \boldsymbol{4!/8!}\\%4
    0\\%5
    0\\%6
    \vdots%7
    \end{matrix} \right) +  (4!)^2   \mL\left( \begin{matrix}
    0\\%0
    \boldsymbol{0}\\%1
    \boldsymbol{0}\\%2
   \boldsymbol{0}\\%3
    \boldsymbol{1}\\%4
    0\\%5
    0\\%6
    \vdots%7
    \end{matrix} \right) \mL\left( \begin{matrix}
    0\\%0
    \boldsymbol{0}\\%1
    \boldsymbol{0}\\%2
    \boldsymbol{0}\\%3
    \boldsymbol{1}\\%4
    0\\%5
    0\\%6
    \vdots%7
    \end{matrix} \right) \mL\left( \begin{matrix}
    0\\%0
   \boldsymbol{1!/9!}\\%1
   \boldsymbol{2!/10!}\\%2
    \boldsymbol{3!/11!}\\%3
    \boldsymbol{4!/12!}\\%4
    0\\%5
    0\\%6
    \vdots%7
    \end{matrix} \right) \,  ,
\end{equation}
where the orders $1,\dots,2\chi$ have been visually highlighted for clarity. 
 In this example, the eleventh order in $t$ is revealed by multiplying the corresponding terms from the Taylor expansion of $\mL$ in Eq.~\eqref{eq:decomp0}, with the corresponding coefficients in \eqref{eq:LCUv2_idea}
\begin{equation}
    (4!)^2 \times \frac{(-\ii Ht)^4}{4!} \times \frac{(-\ii Ht)^4}{4!}  \times   3!/11! \times \frac{(-\ii Ht)^3}{3!} = \frac{(-\ii H t)^{11}}{11!} \, .
\end{equation}
The first term of Eq.~\eqref{eq:LCUv2_idea}  takes care of the zeroth and the first $2\chi$ order of $U$. The second term multiplies all terms of orders $t^1 \dots t^{2\chi}$ with $t^{2\chi}$ and so takes care of the next $2\chi$ terms of the expansion. The third term multiplies with a $\mathcal{O}(t^{2\chi})$ term twice, taking care of the subsequent $2 \chi$ terms -- a pattern emerges. 
In a general setting  (with arbitrary $2 \chi$ and $R$) we can write this sum as
\begin{align}
\notag
    \widetilde{M}_{2\chi,R}^{(\mathrm{cf})}(t)&\coloneqq  \sum_{r=1}^R\left(\mL_{2\chi,R}\left(\nn^{(0)},\bb^{(0)}, t\right)\right)^{r-1}\mL_{2\chi,R}\left(\nn^{(r)},\bb^{(r)}, t\right) \\
    &= \sum_{r=1}^R\left(\sum_{j=0}^{2\chi R} \nu_{j}^{(0)} \frac{\left(-\ii  H t\right)^j}{j!}+ \mathcal{O}\left( t^{2\chi R + 1} \right)\right)^{r-1}\left( \sum_{k=0}^{2\chi R} \nu_{k}^{(r)} \frac{\left(-\ii  Ht \right)^k}{k!}  + \mathcal{O}\left( t^{2\chi R + 1} \right)\right)   \, 
    \label{eq:LCU_2_again} \\
    &= \sum_{r=1}^R \left(\frac{(-\ii H t)^{2\chi}}{(2\chi)!}\right)^{r-1}  \sum_{k=0}^{2\chi R}  \nu_{k}^{(r)} \frac{\left(-\ii  Ht \right)^k}{k!} \; + \; \mathcal{O}\left( t^{2\chi R + 1} \right)  \label{eq:LCU_2_again_and_again}
\end{align}
where we have used Eq.~\eqref{eq:decomp2} in Eq.~\eqref{eq:LCU_2_again} and $\nu_{j}^{(0)} = \delta_{2\chi, j}$ from Definition~\ref{def:our_lcu_v2} to collapse the first factor into $(-\ii H t)^{2\chi(r-1)}/(2\chi)!$ in Eq.~\eqref{eq:LCU_2_again_and_again}.  Considering also that $\nu_{j}^{(0)} = 1$ for $0\leq j \leq 2\chi$, the $r=1$ term (that we recognize are the first $2 \chi$ orders of the time evolution) can be separated from the sum. Discarding all terms that vanish due to $\nu_{k}^{(r)} = 0$  we rewrite Eq.~\eqref{eq:LCU_2_again_and_again} to
\begin{align}
    \widetilde{M}_{2\chi,R}^{(\mathrm{cf})}(t)\; &= \;  \sum_{j=0}^{2\chi} \frac{(-\ii H t)^{j}}{j!} \; + \; \sum_{r=2}^{R} \sum_{k=1}^{2\chi}  \nu_{k}^{(r)} \frac{\left( -\ii Ht \right)^{2\chi(r-1) + k}}{\left((2\chi)!\right)^{r-1}k! } \; + \; \mathcal{O}\left( t^{2\chi R + 1} \right) \,  , \label{eq:scaling2}
\end{align}
for which we consult Definition~\ref{def:our_lcu_v2}, a last time resolving the remaining $\nu_{k}^{(r)}$. This leaves us with the correct time evolution up to order $2\chi R + 1$, thus proving the definition.

\subsection{Error of the averaged operators}
\label{subsec:errors}
 Following upon the insights from Section~\ref{subsec:mpf_1} we find   Corollary~\ref{cor:error_scaling} already proven by Eq.~\eqref{eq:scaling1} and Eq.~\eqref{eq:scaling2} as long as $\widetilde{M}_{2\chi,R}$ is constructed according to Definition~\ref{def:our_lcu_v1} or \ref{def:our_lcu_v2}.
The proof of Theorem~\ref{theo:error_bound} requires a more involved error analysis.
\begin{proof}[Proof of Theorem~\ref{theo:error_bound}]
We first need to bound the remainder terms of the Taylor series expansions of $U(t)$ and $\widetilde{M}^{(\mathrm{m})}$. 
In the following, $\rest_\ell(f)$ denotes the remainder term of the Taylor series of an operator-valued function $f$ truncated at $\ell$-th order in $t$.
We thus find that
\begin{align}
    \nonumber \Norm{\vphantom{\prod_r^R}U(t) - \widetilde{M}^{(\mathrm{m})}_{2\chi,R}} &= 
    \Norm{\ee^{-\ii t H} - \prod_{r=1}^R\left[\sum_{q=0}^{2\chi R} C_q\!\left(\nn^{(r)}, \bb^{(r)}\!\right)\, S_{2\chi}\left( b^{(r)}_qt\right)\right]}\\
    &\leq \Norm{ \vphantom{\rest_{2\chi R}\left(\prod_{r=1}^R\left[\sum_{q=0}^{2\chi R} C_q\!\left(\nn^{(r)}, \bb^{(r)}\!\right)\, S_{2\chi}\left( b^{(r)}_qt\right)\right]\right)}\rest_{2\chi R}\left( \ee^{-\ii t H}\right)} + \Norm{\rest_{2\chi R}\left(\prod_{r=1}^R\left[\sum_{q=0}^{2\chi R} C_q\!\left(\nn^{(r)}, \bb^{(r)}\!\right)\, S_{2\chi}\left( b^{(r)}_qt\right)\right]\right)}.
    \label{eq:bound1}
\end{align}
Moving forward, we employ some recently established results on the theory of Trotter errors.
Specifically, we make use of the `Trotter error with 1-norm scaling' lemma of
Ref.~\cite{childsTheoryTrotterError2020}. Building upon these insights in this fresh context, we find that the exponential remainder of a product formula such as in Eq.~\eqref{eq:proform} can be bounded by
\begin{align}
    \Norm{\rest_\ell\left( \prod_{j=1}^N \ee^{-\ii \alpha_j h_{k_j} t}\right)} &= \frac{t^{\ell+1}}{(\ell+1)!}\Norm{ \left(\frac{\partial}{\partial t}\right)^{\ell+1} \prod_{j=1}^{N} \ee^{-\ii \alpha_j h_{k_j} t}} \\
    &=\frac{t^{\ell+1}}{(\ell+1)!} \Norm{\sum_{x_1 + \dots  + x_N = \ell + 1} \frac{(\ell + 1)!}{x_1! \cdots x_N!} \prod_{j=1}^{N} \left(\frac{\partial}{\partial t} \right)^{x_j} \ee^{-\ii \alpha_j h_{k_j}t}}\nonumber \\
     &=\frac{t^{\ell+1}}{(\ell+1)!} \Norm{\sum_{x_1 + \dots  + x_N = \ell + 1} \frac{(\ell + 1)!}{x_1! \cdots x_N!} \prod_{j=1}^{N} \left(-\ii \alpha_j\,  h_{k_j} \right)^{x_j} \ee^{-\ii \alpha_j h_{k_j} t}}\nonumber\\
      &\leq\frac{t^{\ell+1}}{(\ell+1)!} \sum_{x_1 + \dots  + x_N = \ell + 1} \frac{(\ell + 1)!}{x_1! \cdots x_N!} \prod_{j=1}^{N} \Norm{\alpha_j\,  h_{k_j}}^{x_j} \cdot \underbrace{\Norm{\ee^{-\ii \alpha_j h_{k_j} t}}}_{=1}\nonumber\\
      &= \frac{\left(\sum_{j=1}^N \Norm{\alpha_j \, h_{k_j}} t \right)^{\ell + 1 }}{(\ell+1)!} \, .\nonumber
\end{align}
This bound can now be applied to \eqref{eq:bound1}. For the first term, we obtain
\begin{equation}
    \Norm{\rest_{2\chi R}\left( \ee^{-\ii t H}\right)}\leq \frac{\left(\Lambda t\right)^{2\chi R +1}}{(2\chi R+1)!},
\end{equation}
while the second term can be bounded by

\begin{align}
    \nonumber
    \Norm{\rest_{2\chi R}\left(\prod_{r=1}^R\left[\sum_{q=0}^{2\chi R} C_q^{(r)} S_{2\chi}\left( b^{(r)}_qt\right)\right]\right)}
    &\leq \sum_{q_1,\ldots,q_R=0}^{2\chi R} \Abs{C_{q_1}^{(1)}\cdots  C_{q_{R}}^{(R)}}\Norm{\rest_{2\chi R}\left(S_{2\chi}\left( b^{(1)}_{q_1}t\right)\cdots S_{2\chi}\left( b^{(R)}_{q_{R}}t\right)\right)}\\
    &\leq \frac{\left(g_\chi\Lambda t\right)^{2\chi R+1}}{(2\chi R+1)!}\sum_{q_1,\ldots,q_R=0}^{2\chi R} \Abs{C_{q_1}^{(1)} \cdots  C_{q_{R}}^{(R)}}
    \left(\sum_{i=1}^R\abs{b_{q_i}^{(i)}}\right)^{2\chi R+1}
\end{align}

where we have introduced the factor
\begin{equation}
    g_\chi \coloneqq\frac{4\chi}{3}\left(\frac{5}{3}\right)^{\chi-1}
\end{equation}
to group all terms relating to $\chi$ from the Trotter-Suzuki decomposition appearing in Eq. (53) of Ref.~\cite{childsHamiltonianSimulationUsing}.

At the same time, we have defined $\Lambda\coloneqq \sum_k\norm{h_k}$.
Consequently, we can bound the error of the matching multi-product formula via
\begin{equation}
    \Norm{U(t)-\widetilde{M}^{(\mathrm{m})}_{2\chi,R}(t)} \leq \left(1 + \zeta^{(\mathrm{m})}_{{\chi,R}} g_{\chi}^{2\chi R+1}\right) \frac{\left( \Lambda t \right)^{2\chi R+1}}{\left(2\chi R+1\right)!} \,, 
\end{equation}
with
\begin{equation}
\label{eq:zeta_m}
    \zeta^{(\mathrm{m})}_{{\chi,R}} \coloneqq \sum_{q_1=0}^{2\chi R}\sum_{q_2=0}^{2\chi R}\cdots \sum_{q_R=0}^{2\chi R} \Abs{C_{q_1}^{(1)}C_{q_2}^{(2)} \cdots  C_{q_R}^{(R)}}\left( \abs{b_{q_1}^{(1)}}+\ldots +\abs{b_{q_R}^{(R)}}\right)^{2\chi R+1}
\end{equation}
The error analysis of the closed-form multi-product formula follows along similar lines: Again, we bound the remainder terms of the Taylor series expansions of $U(t)$ and $\widetilde{M}^{(\mathrm{cf})}$, finding
\begin{align}
    \nonumber &\Norm{U(t) - \widetilde{M}^{(\mathrm{cf})}_{2\chi,R}} = 
    \Norm{\ee^{-\ii t H} - \sum_{r=1}^R\left(\mL_{2\chi,R}\left(\nn^{(0)},\bb^{(0)}, t\right)\right)^{r-1}\mL_{2\chi,R}\left(\nn^{(r)},\bb^{(r)}, t\right)}
    \\
    &\leq \Norm{ \vphantom{\rest_{2\chi R}\left(\sum_{r=1}^R\left[\sum_{q=0}^{2\chi R} C_q\!\left(\nn^{(r)}, \bb^{(r)}\!\right)\,   S_{2\chi}\left( b^{(r)}_qt\right)\right]^{r-1}\left[\sum_{q=0}^{2\chi R} C_q^{(r)} \, S_{2\chi}\left( b_qt\right)\right]\right)}    \rest_{2\chi R}\left( \ee^{-\ii t H}\right)}  \notag  \\ & \quad +  \; \Norm{\rest_{2\chi R}\left(\sum_{r=1}^R\left[\sum_{q=0}^{2\chi R} C_q\!\left(\nn^{(0)}, \bb^{(0)}\!\right)\, S_{2\chi}\!\left( b_q^{(0)}t\right)\right]^{r-1}\left[\sum_{q=0}^{2\chi R} C_q\!\left(\nn^{(r)}, \bb^{(r)}\!\right)\, \, S_{2\chi}\!\left( b_q^{(r)}t\right)\right]\right)} \, ,
    \label{eq:bound2}
\end{align}
where the second term can be bounded by
\begin{align}
    &\Norm{\rest_{2\chi R}\left(\sum_{r=1}^R\left[\sum_{q=0}^{2\chi R} C_q\!\left(\nn^{(0)}, \bb^{(0)}\!\right)\, S_{2\chi}\!\left( b_q^{(0)}t\right)\right]^{r-1}\left[\sum_{q=0}^{2\chi R} C_q\!\left(\nn^{(r)}, \bb^{(r)}\!\right)\, \, S_{2\chi}\!\left( b_q^{(r)}t\right)\right]\right)} \notag \\
    & \; \leq \; \sum_{r=1}^R\sum_{q_1=0}^{2\chi R}\cdots \sum_{q_r=0}^{2\chi R} \Abs{C_{q_1}(\nn^{(0)}, \bb^{(0)}) \cdots C_{q_{r-1}}(\nn^{(0)}, \bb^{(0)}) \, C_{q_r}(\nn^{(r)}, \bb^{(r)})} \notag\\ & \qquad\qquad\qquad\qquad\qquad \times \Norm{\rest_{2\chi R}\left(S_{2\chi}\left(b_{q_1}^{(0)}\right) \cdots S_{2\chi}\left(b_{q_{r-1}}^{(0)}\right)\, S_{2\chi}\left(b_{q_r}^{(r)}\right)\right)}
\end{align}
Defining
\begin{eqnarray}
\label{eq:zeta_cf}
\zeta^{(\mathrm{cf})}_{{\chi,R}}&\coloneqq& \sum_{r=1}^R\sum_{q_1=0}^{2\chi R}\sum_{q_2=0}^{2\chi R}\cdots \sum_{q_r=0}^{2\chi R} \Abs{C_{q_1}^{(0)}C_{q_2}^{(0)} \cdots  C_{q_{r-1}}^{(0)}C_{q_{r}}^{(r)}}\left( \abs{b_{q_1}^{(0)}}+\ldots+\abs{b_{q_{r-1}}^{(0)}} +\abs{b_{q_r}^{(r)}}\right)^{2\chi R+1},\\
\Lambda&\coloneqq & \sum_k\norm{h_k}\, ,
\end{eqnarray}
we find
\begin{equation}
    \Norm{U(t)-\widetilde{M}^{(\mathrm{cf})}_{2\chi,R}(t)} \leq \left(1 + \zeta^{(\mathrm{cf})}_{{\chi,R}} g_{\chi}^{2\chi R+1}\right) \frac{\left( \Lambda t \right)^{2\chi R+1}}{\left(2\chi R+1\right)!} \,, 
\end{equation}
concluding the proof.
\end{proof}

\section{Comparison and numerical validation}
\label{sec:numerics}

We will now compare the Trotter-Suzuki product formula algorithm and multi-product formulas in the randomized sampling framework.
To make comparisons as fair as possible, we will assume that each algorithm uses at most $R$-fold sequences of $S_{2\chi}(\cdot)$ blocks. Product formulas may use repetitions as in Eq.~\eqref{eq:reps} achieving an error of $\mathcal{O}((t/R)^{2\chi + 1 })$, which is a reliable way to approximate time evolutions for longer times $R \gg \tau > 1$. However, repetitions improve the accuracy of shorter time evolutions only minimally. The situation is different if $R$ is an integer power of five, which allows us to build the next higher Trotter-Suzuki order and approximate the time evolution up to a leading order of $2(\chi + \log_5 R) + 1$ in $t$. This means improving the leading power of $t$ comes at an exponential cost for the circuit depth. The situation can be remedied by sampling from a multi-product formula. Remarkably, the statement for sampling observable eigenvalues with product formulas is very similar to Theorem~\ref{theo:time_evolution}. 

Since the sampling errors of all of these methods are comparable, we disregard them in this comparison and instead only think about asymptotic limits. The number of required circuit evaluations for the same accuracy then follows directly from the corresponding resolution factors.
The main difference between product and multi-product formulas is that Trotter-Suzuki algorithms do not have a resolution factor. This is equivalent to $\Xi=1$ with consequences for sampling complexity and error bounds. This advantage 
is, however, quickly outweighed as Childs and Wiebe's multi-product formula delivers an improved approximation that is exact up to a leading order of $2(\chi + R)-1$ in $t$. Modifying the leading power of $t$ is now possible by adding exponentially fewer terms. Here, we consider only the simplest form of Childs and Wiebe's multi-product formula with $\ell_q=q$ since their improvements, e.g. those proposed in Ref.~\cite{lowWellconditionedMultiproductHamiltonian2019}, focus on increasing the success probability of their deterministic implementation without improving their error scaling and lead to significantly worse resolution factors. 
Furthermore, note that they require the same circuit depth for $R=K+1$ in the randomized sampling framework due to the final term in Eq.~\eqref{eq:multi_product_formula}. With matching and closed-form multi-product formulas, the approximation can be further improved to leading order $2\chi R+1$ in $t$. This means that only one additional  $S_{2\chi}(\cdot)$ block improves the order by $2\chi$ rather than $2$ as for multi-product formulas of the prior art.
While the scaling in $t$ of the novel multi-product formulas is superior to those of Childs and Wiebe, their theoretical error bound is not as tight, leading to a crossover of error bounds at $\sum|h_k| t<1$. While it is true that the error bounds cross over at comparably small errors, it is important to note that the significantly reduced resolution factors allow for multiple repetitions, i.e., $r$ evolutions of shorter times $t/r$, until the same number of circuit evaluations is required as for Childs and Wiebe's formulas.

We have plotted the error bounds of all formulas for fixed circuit depth between all methods for one particular optimization of the multi-product formula parameters in Fig.~\ref{fig:compare_bounds}. Optimizing the set of parameters $\lbrace\boldsymbol{b}^{(r)} \rbrace$, we tend to achieve noticeably lower resolution factors than obtained using Childs and Wiebe's multi-product formula. We generally find the matching multi-product formula to have a smaller resolution factor than the closed-form multi-product formula. While the matching multi-product formula has a better resolution, it requires an additional layer of classical optimization. 

\begin{figure}[t!]
    \centering
    \includegraphics[scale=0.875]{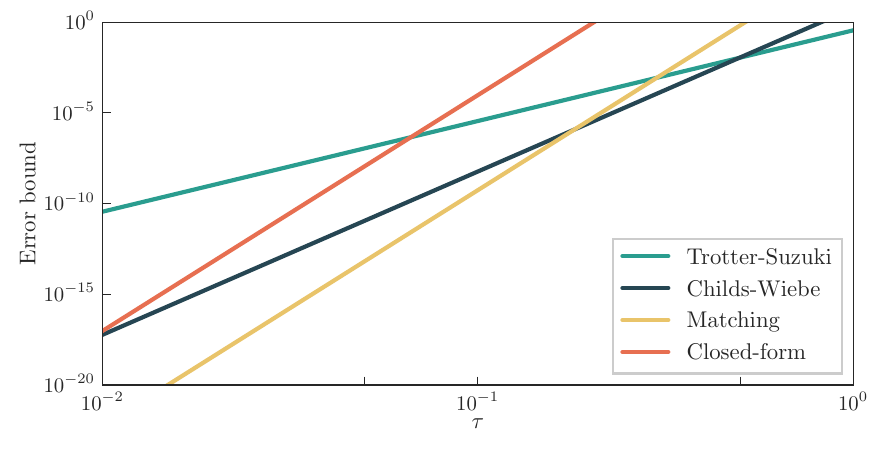}
    \caption{Comparing error bounds for Trotter-Suzuki product formulas as well as several multi-product formulas for different $\tau=\Lambda t$. 
    The exact formulas for these bounds are provided in Equations~\eqref{eq:trotter_bound} and \eqref{eq:childs-wiebe_bound} and Theorem~\ref{theo:error_bound}.
    We have chosen $2\chi = 4$ and $r=R=K+1=3$ to fix the depth of all methods to be 30 oracle calls. The parameters $\lbrace \boldsymbol{b}^{(r)}\rbrace$ have been numerically optimized for the matching and closed-form multi-product formulas with an initial guess of $\boldsymbol{b}=\set{1,-1,2,-2,\ldots,7}$, while we employ $\ell_q=q$ for Childs and Wiebe's formula as the choice with the lowest resolution factor. We find the resolution factors for the multi-product formulas to be $\Xi^{(\mathrm{CW})}\! \approx 3.13$, $\Xi^{(\mathrm{m})} \!\approx 1.22$ and $\Xi^{(\mathrm{cf})} \!\approx 1.36$. 
     Consequently, the matching and closed-form multi-product formulas would allow for four to five repetitions before they require the same amount of circuit evaluations as Childs and Wiebe's formula.
     In the black box optimization, we have used the error bound as the objective function with the modification of using 
    $(\Xi^\mathrm{(m)})^{20}$ and $(\Xi^\mathrm{(cf)})^{10}$, instead of $\zeta_{\chi,R}$ respectively, to ensure reasonable resolution factors.}
    \label{fig:compare_bounds}
\end{figure}

It is important to note that since we did not find useful bounds relating $\boldsymbol{b}$ to any of the resolution factors, further improvements of the bounds in Fig.~\ref{fig:compare_bounds} seem possible. 
In the absence of relations $\Xi(b)$, it is necessary to numerically optimize all $\boldsymbol{b}^{(r)}$ for a chosen loss function, which is not the case for Childs and Wiebe's multi-product formula due to existing, analytical relations.
We found that global optimization using basin hopping combined with Nelder-Mead optimization yields the best results since the optimization landscape exhibits a large number of local minima. 
The optimization is also sensitive to the initial guess for those parameters, with an equal spread of positive and negative integers, i.e., $\boldsymbol{b}_\text{init}=(1,-1,2,-2,3,-3,\ldots,\chi R+1)$, leading to better results. It is also fruitful to vary the loss function used for optimization. 
On the one hand, using the error bound for a fixed $\tau$ yields the lowest error, with the downside that the corresponding resolution factors are too large to be practical. 
On the other hand, solely optimizing for the resolution factor might result in resolution factors arbitrarily close to one with the downside of significantly worse error bounds. 
We, therefore, found the most fruitful loss function to be the error bound with the modification of using the resolution factor $\Xi$ directly instead of $\gamma$ and amplifying its impact on the loss function by using $\Xi^p$ for some power $p$, which allows us to balance both quantities. Nevertheless, the use of different loss functions could be explored further. We have also found that bounding each parameter $b$ by the total maximum leads to more robust results.

To corroborate the functioning of the newly proposed multi-product formulas beyond 
theoretical bounds, we also compare the actual performance with that of the conventional multi-product formula and Trotter-Suzuki. Here, we compare the actual operator distance to the ideal time evolution for the following five physically plausible and meaningful Hamiltonians, with the results shown in Fig.~\ref{fig:compare_numeric}:
\begin{figure}[t!]
    \centering
    \includegraphics[width=.95\linewidth]{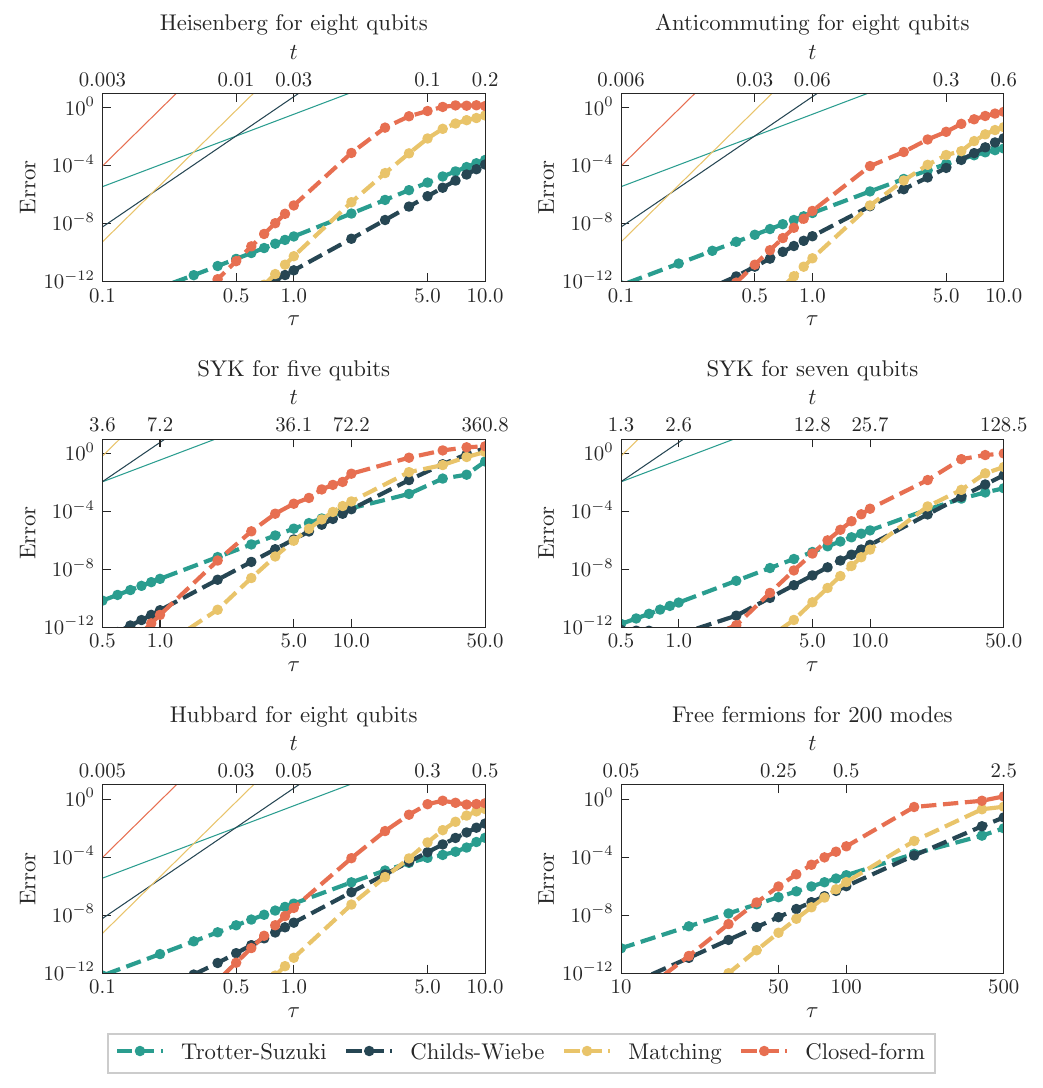}
    \caption{Numerical comparison of Trotter-Suzuki product formulas, Childs and Wiebe's multi-product formula, and the multi-product formulas proposed in this work for different $\tau=\sum_k||h_k|| t$. Here, we approximate the time evolution operator and plot the operator distance between the approximation and the ideal evolution operator for several physically plausible and interesting local, but not necessarily geometrically local, Hamiltonians, defined in equations \eqref{eq:heis}, \eqref{eq:anticom}m \eqref{eq:syk}, \eqref{eq:hubbard} and \eqref{eq:fermions}. 
    The simulated distances (thick, dashed lines) are much smaller than the theoretical bounds (thin, solid lines), but do have remarkably similar features overall.
    }
    \label{fig:compare_numeric}
\end{figure}
First, we consider a standard \emph{Heisenberg Hamiltonian} with periodic boundary conditions, described by

\begin{equation}
\label{eq:heis}
    H_\mathrm{Heisenberg} = -\sum_{\langle i,j\rangle}\left(X_iX_j + Y_iY_j + Z_iZ_j\right) + 2 \sum_i X_i.
\end{equation}
This is a Hamiltonian that plays an important role in condensed matter physics, as a prototypical model capturing ferromagnetism.
Trotter-Suzuki and Childs-Wiebe formulas perform extremely well, as the Hamiltonian comprises many commuting terms. Since their performance guarantees rely on nested commutators \cite{childsTheoryTrotterError2020}, this behavior is to be expected and demonstrates the superior performance of these well studied formulas for lattice Hamiltonians, which is not reached by the newly proposed multi-product formulas. However, they are less optimal for Hamiltonians with fewer commuting terms.
Motivated by these findings, we now turn to investigate a Hamiltonian comprising mutually anti-commuting terms, defined as
\begin{equation}
\label{eq:anticom}
H_\mathrm{anti} = \sum_{i=0}^{7} Z^{\otimes i}\otimes (X+Y) + Z^{\otimes 8}.
\end{equation}
Anti-commuting terms play an important role when using quantum
simulation algorithms to investigate \emph{fermionic quantum models}.  Here, we find a notable and significant advantage of the newly proposed multi-product formulas over Trotter-Suzuki and Childs-Wiebe already for quite large parameters of $\tau = \sum|h_k|t$. Consequently, we expect our formulas to work well and exceed previous methods for fermionic systems, which, once transformed into qubit Hamiltonians using the Jordan-Wigner transformation, will have fewer commuting terms than standard lattice spin Hamiltonians.

This behavior is confirmed and further corroborated by our results of the
\emph{Sachdev–Ye–Kitaev (SYK)} 
model as defined in Ref.~\cite{babbushQuantumSimulationSachdevYeKitaev2019}, whose Hamiltonian is given by

\begin{equation}
\label{eq:syk}
H_\mathrm{SYK} = \frac{1}{4\cdot4!}\sum_{p,q,r,s=0}^{N-1} J_{p,q,r,s} \gamma_p\gamma_q\gamma_r\gamma_s,
\end{equation}
where $N$ is the number of Majorana fermion mode operators $\gamma_p$ and the $J_{p,q,r,s}$ are real-valued scalars drawn randomly from a normal distribution with variance $\sigma^2=3!/N^3$. This is an intricate 
local, but not geometrically local, model that is 
believed to provide insights into
instances of strongly correlated quantum 
materials. It is used in the study of scrambling
dynamics and has a close relation with discrete models that capture aspects of holography in the black hole context.

Again invoking the Jordan-Wigner transformation, $N$ Majorana fermion mode operators can be mapped onto $N/2$ qubits. For our simulations, we thus chose $N=10$ and $N=14$, leading to a five and seven-qubit model, respectively. Here, we again find a notable 
and in instances substantial advantage for 
large $\tau$ and even for large times $t$. This is a physically highly plausible and interesting model for which our new simulation methods fare well.

As a further fermionic system, we look at the spinful Hubbard model on two by two sites, defined in Ref.~\cite{mccleanOpenFermionElectronicStructure2020} as 

\begin{equation}
\label{eq:hubbard}
H_\mathrm{Hub} = - t \sum_{\langle i,j \rangle, \sigma}(a^\dagger_{i, \sigma} a_{j, \sigma} + a^\dagger_{j, \sigma} a_{i, \sigma}) + U \sum_{i} a^\dagger_{i, \uparrow} a_{i, \uparrow} a^\dagger_{i, \downarrow} a_{i, \downarrow} - \mu \sum_i \sum_{\sigma} a^\dagger_{i, \sigma} a_{i, \sigma}- h \sum_i (a^\dagger_{i, \uparrow} a_{i, \uparrow} -  a^\dagger_{i, \downarrow} a_{i, \downarrow}),
\end{equation}
with spin $\sigma$, tunneling amplitude $t=2$, Coulomb potential $U=2$, magnetic field $h=0.5$ and chemical potential $\mu=0.25$. Also, $a$ and $a^\dagger$ represent fermionic annihilation and creation operators. These are comparably small system sizes, but already
show the substantial potential of the proposed simulation method.

Finally, as a last family of examples, in order
to gauge the performance of our methods for larger system sizes, we investigate a system of 200 non-interacting (``free'') 
fermions with nearest neighbor interactions and periodic boundary conditions

\begin{equation}
\label{eq:fermions}
H_\mathrm{ff}(h) = \sum_{i,j}h_{i,j}a^\dagger_i a^\dagger_j,
\end{equation}
with $h_{i,j} = 1$ if the respective fermions are nearest neighbors and $h_{i,j} = 0$ otherwise.
Note that for gauging the performance for free fermions, we do not compare the time evolution operator $U(t)=\ee^{-\ii Ht}$ to its approximation, but the Greens function propagator $G(t)=\ee^{-\ii ht}$ to its approximation.

All comparisons are done in the randomized sampling framework for a Trotter-Suzuki order of $2\chi=4$, the number of repetitions $R=3$ and corresponding parameters for all other algorithms, ensuring an equal depth measured in the number of the required oracle calls. Although the straightforward comparison of their bounds in Fig.~\ref{fig:compare_bounds} suggests an advantage of our multi-product formulas at about $\tau=0.1$, we find that this is the case for much larger $\tau$ already on actual systems as shown in Fig.~\ref{fig:compare_numeric}.
Note also that since the Hamiltonians we consider are not necessarily geometrically local, known classical simulation techniques will be heavily challenged even for comparably short simulation times

While the performance of the simple multi-product formula and Trotter-Suzuki is also much better than their bounds indicate, we find that the presented bounds for our proposed multi-product formulas are comparably looser. It is also worth noting that the performance of the matching version is slightly better than that of the closed-form version, although it comes with the penalty of an additional, classical optimization loop.
We also find that the advantages of our multi-product formulas are visible even for $\tau>1$; in the case of the SYK model,
this holds even for actual simulation times $t>1$. Additionally, we find that this advantage can also be maintained for larger system sizes, as indicated by their performance on the model of free fermions.
The presented numerical studies, therefore, provide strong arguments for the functioning of our proposed multi-product formulas and their advantage in the presented regimes.
It is also important to note that the corresponding resolution factors required for the randomized sampling scheme of $\Xi^\mathrm{(cf)}\approx 1.36$ and $\Xi^\mathrm{(m)}\approx 1.22$ are significantly better than the $\Xi^\mathrm{(CW)}\approx 3.13$ of the multi-product formula proposed by Childs and Wiebe in Ref.~\cite{childsHamiltonianSimulationUsing} and would thus allow for at least four repetitions before they require the same overhead in circuit evaluations.

\section{Discussion and conclusion}

In this work, we have brought together two main ingredients of methods of quantum simulation. The results are notably more resource-efficient ways of performing short-time Hamiltonian simulation. These are on the one hand higher-order multi-product formulas \cite{childsHamiltonianSimulationUsing,lowWellconditionedMultiproductHamiltonian2019}, on the other an element which has long been underappreciated but recently been of high interest: the element of \emph{randomness} \cite{campbellShorterGateSequences2017b,campbellRandomCompilerFast2019a,childsFasterQuantumSimulation2019,ouyangCompilationStochasticHamiltonian2020,chenQuantumSimulationRandomized2020,chenQuantumSimulationRandomized2020}. Overcoming the prejudice that the time evolution has to be completed in each run of a compilation, we have
introduced a novel framework for implementing multi-product formulas, to estimate expectation values of time-evolved observables.

Concretely, we have proposed a randomized sampling approach that focuses on the time evolution of the observable, not the state.
When implementing multi-product formulas in a randomized fashion rather than via block encodings in the LCU framework, we can circumvent the need for additional amplitude amplification or post-selection. The results presented here have been obtained by only requiring access to a quantum-oracle machine that implements single-qubit state preparation, controlled time evolution, and quantum measurements. They are thus especially relevant in the regime of early quantum computers in which NISQ algorithms reach their limits but where full-fledged, digital, long-time evolution algorithms on fault-tolerant quantum computers are not yet available. Consequently, this work may be seen as targeting a regime in between the digital and analog setting, where we have some form of parametric control over a simulator system allowing us to compile the target time evolution with sequences of the simulator’s time evolutions \cite{Trotzky,parra-rodriguezDigitalanalogQuantumComputation2020}. This programmable regime then constitutes a departure from the analog setting with relatively little control over the simulator and is not as strict as digital simulation, where the control over the quantum system is strong enough to fashion its interactions into quantum gates.

Within this randomized sampling framework, we have proposed two new multi-product formulas. These schemes have been equipped with full rigorous performance guarantees. Furthermore, we have included a detailed estimation of the number of circuit evaluations that are required, a vital metric for randomized approaches. Comparing the error bounds of these newly introduced algorithms with Trotter-Suzuki product formulas and Childs and Wiebe's multi-product formula, we find that they outperform the latter for a fixed circuit depth in a practically relevant regime.

The use of multiple repetitions of short time evolutions by $t/r$ instead of a single, long time evolution by $t$, elementary to achieving long simulation times, comes with a penalty in the number of circuit evaluations which scales exponentially in $r$. However, the base of this exponential, the resolution factor in our case, is only slightly larger than one and could be optimized even further. Even for our toy examples, it was in the range of $1.2-1.4$ for the newly proposed multi-product formulas as compared to $3.13$ for Childs and Wiebe type multi-product formulas. Thus, at least a few repetitions are still in reach, rendering these results especially relevant for early quantum computers since we do not solely consider geometrically local Hamiltonians, for which even comparably short simulation times are a highly difficult task for known classical simulation techniques.

Benchmarking them on five different Hamiltonians, all of which are physically well motivated and each interesting in its own right, we have found that this advantage can be expected already at comparably large simulation times. While lattice Hamiltonians with many commuting terms are most likely best approximated using Trotter-Suzuki or Childs-Wiebe formulas, the newly proposed multi-product formulas show a clear advantage for fermionic Hamiltonians and those with a small number of commuting terms. This insight points to the direction that there might not be a universally optimal quantum simulation algorithm for digital quantum simulation. Instead, some algorithms could be better suited to capture the specifics of a given local Hamiltonian model. The downside of the proposed methods is that more measurements are required to reach the desired precision through the resolution factor. This resolution factor can be optimized using a classical black-box optimization, which is de facto a requirement for the functioning of the proposed multi-product formulas.  

The present work is essentially bridging the gap between analog and (perhaps error-corrected), fully digital quantum technology: Not only do we expect there to be other randomized sampling schemes in digital quantum simulation, but once one can replace the element of randomness with block encodings, one can switch from these expectation value based algorithm to algorithms based on quantum phase estimation. Overall, the method introduced gives rise to a less resource-demanding way of performing Hamiltonian simulation,
while also remaining conceptually and technologically
simpler than for instance qubitization
\cite{lowHamiltonianSimulationQubitization2019b},
bringing such ideas to an extent closer to the resources available in early quantum computers.

Looking ahead, it remains an open problem to relate the parameters $\boldsymbol{b}$ in Definition~\ref{def:new_multi_prod_base} to the resolution factor in a way that would eliminate the need for black-box optimization. So far, we can only connect the two quantities analytically, and that involves the complicated product with the Vandermonde matrix~\eqref{eq:vandermonde}. Alternatively, one could improve the optimization rather than replace it.  We have used only simple optimizers and loss functions, and expect possible improvements for more involved loss functions and optimization algorithms.
Furthermore, the presented constructions are just two of the plethora of new multi-product formulas that could be constructed with Definition~\ref{def:new_multi_prod_base} and might exhibit better error bounds and resolution.

\section{Acknowledgments}
This work has been supported by the DFG (CRC 183 project B01 and A03, EI 519/21-1). This work has also received funding from the European Unions Horizon 2020 research and innovation program under grant agreement No.\ 817482 (PASQuanS), specifically dedicated to programmable quantum simulators. It has also been supported by the BMWK (PlanQK and EniQmA), the BMBF (DAQC on notions of digital-analog quantum simulation and FermiQP on fermionic quantum processors), the Munich Quantum Valley (K8), and the Einstein Foundation (Einstein Research Unit on Quantum Devices) .
M.~K.~acknowledges funding from ARC Centre of Excellence for Quantum Computation and Communication Technology (CQC2T), project number CE170100012. 
The authors endorse Scientific CO$_2$nduct \cite{Sweke_2022} and provide a CO$_2$ emission table in the appendix.

\bibliographystyle{quantum}
\bibliography{literature}

\begin{appendix}
\section{CO\texorpdfstring{$_2$}{2} emission table}
    \begin{center}
    \begin{tabular}[b]{l c}
    \hline
    \textbf{Numerical simulations} & \\
    \hline
    Total Kernel Hours [$\mathrm{h}$]& $\approx800$\\
    Thermal Design Power Per Kernel [$\mathrm{W}$]& 5.75\\
    Total Energy Consumption Simulations [$\mathrm{kWh}$] & 4.6\\
    Average Emission Of CO$_2$ In Germany [$\mathrm{kg/kWh}$]& 0.56\\
    Total CO$_2$-Emission For Numerical Simulations [$\mathrm{kg}$] & 2.6\\
    Were The Emissions Offset? & \textbf{Yes}\\
    \end{tabular}
    \end{center}
\end{appendix}
\end{document}

%% file: figures/multi-product-formulas_tikz.tex
    \begin{tikzpicture}[
        squarednode/.style={rectangle, rounded corners, draw=black, fill=black!0.01, very thick, minimum size=5mm, minimum width=1cm, align=center}]
        %Nodes
        \node[squarednode]        (MPF)       {Multi-product formulas};
        \node[squarednode]      (QEC)        [below=.75cm of MPF.south]                      {Quantum error mitigation \cite{endoHybridQuantumClassicalAlgorithms2021}};
        \node[squarednode]      (LCU)       [left=.75cm of QEC] {Linear-combination-of-unitaries \cite{childsHamiltonianSimulationUsing,lowWellconditionedMultiproductHamiltonian2019}};
        \node[squarednode]        (RS)       [right=.75cm of QEC] {Randomized sampling [here]};
        
        %Lines
        \draw[-latex,very thick] (MPF.south) -- (QEC.north);
        \draw[-latex,very thick] (MPF.south) -- (RS.north);
        \draw[-latex,very thick] (MPF.south) -- (LCU.north);
    \end{tikzpicture}

%% file: figures/sampling_algo.tex
\begin{tikzpicture}[
squarednode/.style={rectangle, rounded corners, draw=black, fill=black!0.01, very thick, minimum size=5mm, minimum width=1cm, align=center}
]
\node[squarednode](master){
\begin{tikzpicture}[
squarednode/.style={rectangle, rounded corners, draw=black, fill=black!0.01, very thick, minimum size=5mm, minimum width=1cm, align=center}
]
%Nodes
\node[rectangle, align=flush left] (s1) {\textbf{Step 1:} Prepare multi-product formula for importance sampling:};
\node[rectangle, align=flush left, below= of s1.west, anchor=west] (eq1) {$\displaystyle \qquad e^{-iHt} \approx\sum\limits_{k=1}^M|C_k|\left(\operatorname{sign}(C_k)\widetilde{V}_k\right)\quad \xrightarrow[\text{divide by }\Xi\coloneqq\sum|C_k|]{\text{ensure probability ensemble}}\quad V = \sum\limits_{k=1}^M \underbrace{\frac{|C_k|}{\Xi}}_{p_k}\underbrace{\left(\operatorname{sign}(C_k) \tilde{V}_k\right)}_{V_k} \approx \frac{1}{\Xi}e^{-iHt}$};

\node[rectangle, align=flush left](s2)[below=of eq1.west, anchor=west]{\textbf{Step 2:} Repeat the following two steps $N$ times
 (see Algorithm~\ref{algo:random_sampling}):};
\node[rectangle, align=flush left](a)[below= .75cm of s2.west, anchor=west]{$\qquad$\textbf{(a)} Sample $V_\circ, V_\bullet \stackrel{i.i.d.}{\sim}\{p_k,V_k\}$ via importance sampling, i.e., draw $V_k$ with probability $p_k$.};

\node[rectangle, align=flush left](b)[below=.75cm of a.west, anchor=west]{$\qquad$\textbf{(b)} Run the following circuit and measure outcome $o_j$, with $\mathbb{E}(o_j) = \frac{1}{2}\mathrm{Tr}\left(O(V^{\phantom{\dagger}}_\circ\rho V_\bullet^\dagger +V^{\phantom{\dagger}}_\bullet\rho V_\circ^\dagger)\right)$:};
\node[rectangle, align = center](circuit)[below=.6cm of b.north]{\includegraphics{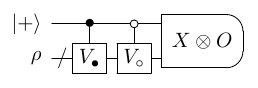}};
\node[rectangle, align = center](placeholder)[below=of b.west, anchor=west]{};
\node[rectangle, align=flush left](res)[below=1.25cm of placeholder.west, anchor=west]{\textbf{Main Result:} For $N\geq \frac{2\lVert O\rVert^2 \ln(2/\delta)}{\epsilon^2}$, we find that $\left\lvert\frac{1}{N}\sum_{j=1}^No_j-\operatorname{tr}(OV\rho V^\dagger)\right\rvert\leq\epsilon$ with confidence $\delta$.};
\end{tikzpicture}};
\end{tikzpicture}

%% file: figures/prescription_tikz.tex
\begin{tikzpicture}[squarednode/.style={rectangle, rounded corners, draw=black, fill=black!0.01, very thick, minimum size=5mm, minimum width=1cm, align=center}
]
%Nodes

\node[squarednode] (matching) [minimum width=5.65cm] {Use the matching multi-product formula\\ introduced in Definition \ref{def:our_lcu_v1}};
\node[squarednode] (closedform) [right=.75cm of matching,minimum width=5.6cm] {Use the closed-form multi-product formula\\introduced in Definition \ref{def:our_lcu_v2}};
\node[squarednode] (attempt) [above = 1.5cm of $(matching)!0.5!(closedform)$]{Attempt solving the system of nonlinear equations defined in Eq.~\eqref{eq:constraint}};
\node[rectangle] (belowmatching)[below = 0.75cm of matching]{};
\node[rectangle] (belowcf)[below = 0.75cm of closedform]{};
\node[squarednode] (optimizing) [below = 0.75cm of $(matching.south)!0.5!(closedform.south)$,minimum width=15.97cm, text width=15.9cm,,align=flush left]{Numerically optimize the parameters $\boldsymbol{b}$ using a loss function balancing error bound (see Theorem \ref{theo:error_bound}) and resolution factor $\Xi=\sum_q|C_q|$ impacting the number of circuit evaluations.

This involves inverting the Vandermonde matrix defined in Eq.~\eqref{eq:vandermonde}};
\node[squarednode] (output) [below=0.75cm of optimizing] {Obtain a multi-product formula with $\sum_{q} C_q V_q\approx \mathrm{e}^{\mathrm{-\mathrm{i} Ht}}$};
\node[squarednode] (randomize) [draw=colorYE, fill=colorYE!5, below=0.75cm of output] {Implement MPF using the newly proposed \emph{randomized sampling} framework (Section~\ref{sec:rand_samp_algo} and Figure \ref{fig:sampling_algo})};
%Lines
\draw[-latex, very thick] (attempt.south)-- node [pos=0.65,above, sloped] {\textcolor{colorGR}{Success}} (matching.north);
\draw[-latex, very thick] (attempt.south)-- node [pos=0.65,above, sloped] {\textcolor{colorRE}{Failure}} (closedform.north);
\draw[-latex, very thick] (matching.south) -- (belowmatching.north);
\draw[-latex, very thick] (closedform.south) -- (belowcf.north);
\draw[-latex, very thick] (optimizing.south) -- (output.north);
\draw[-latex, very thick] (output.south) -- (randomize.north);
\end{tikzpicture}